\documentclass[sigconf,screen]{acmart}
\AtBeginDocument{%
  }

\usepackage{multirow}
\usepackage{makecell}
\usepackage{algorithm}
\usepackage{algorithmic}
\usepackage{siunitx}
\usepackage{CJK}

\newtheorem{lemma}{Lemma}
\newtheorem{theorem}{Theorem}
\newtheorem{corollary}{Corollary}

\setcopyright{acmlicensed}
\copyrightyear{2018}
\acmYear{2018}
\acmDOI{XXXXXXX.XXXXXXX}
\acmConference[Conference acronym 'XX]{Make sure to enter the correct
  conference title from your rights confirmation email}{June 03--05,
  2018}{Woodstock, NY}
\acmISBN{978-1-4503-XXXX-X/2018/06}

\begin{document}

\title[ReTokSync: Tokenization Disambiguation]{ReTokSync: Self-Synchronizing Tokenization Disambiguation for Generative Linguistic Steganography}

\author{Yaofei Wang}
\affiliation{%
  \institution{Hefei University of Technology}
  \city{Hefei}
  \country{China}
}
\email{wyf@hfut.edu.cn}

\author{Rui Wang}
\affiliation{%
  \institution{Hefei University of Technology}
  \city{Hefei}
  \country{China}
}
\email{2025170799@mail.hfut.edu.cn}

\author{Weilong Pang}
\affiliation{%
  \institution{Hefei University of Technology}
  \city{Hefei}
  \country{China}
}
\email{pangwl@mail.hfut.edu.cn}

\author{JiaLiang Han}
\affiliation{%
  \institution{Hefei University of Technology}
  \city{Hefei}
  \country{China}
}
\email{hanjl@mail.hfut.edu.cn}

\author{Yuan Qi}
\affiliation{%
  \institution{University of Science and Technology of China}
  \city{Hefei}
  \country{China}
}
\email{qya7ya@mail.ustc.edu.cn}

\author{Donghui Hu}
\affiliation{%
  \institution{Hefei University of Technology}
  \city{Hefei}
  \country{China}
}
\email{hudh@hfut.edu.cn}

\author{Kejiang Chen}
\affiliation{%
  \institution{University of Science and Technology of China}
  \city{Hefei}
  \country{China}
}
\email{chenkj@ustc.edu.cn}

\renewcommand{\shortauthors}{Wang et al.}

\renewcommand{\shortauthors}{Trovato et al.}

\begin{abstract}

Generative linguistic steganography (GLS) enables covert communication by embedding secret messages into the natural language generation process. In practical deployment, however, GLS is vulnerable to tokenization ambiguity: the same surface text may be re-tokenized into a different token sequence at the receiver, breaking the shared decoding state between the communicating parties so that a single local mismatch can propagate into complete extraction failure. Existing solutions either remove ambiguous tokens---distorting the generation distribution and compromising security---or preserve the distribution at the cost of substantially reduced embedding capacity or prohibitive runtime overhead. To address this issue, we propose ReTokSync (Re-Tokenization Synchronization), a self-synchronizing disambiguation framework that monitors the receiver-view tokenization during generation and triggers a corrective reset only when ambiguity actually occurs. By confining the effect of tokenization ambiguity to sparse residual bit errors rather than global desynchronization, ReTokSync leaves ambiguity-free positions entirely untouched and remains compatible with the underlying steganographic algorithm. Experiments on both English and Chinese settings show that ReTokSync stays closest to the steganographic baseline in distributional security (zero KL divergence), text quality, embedding capacity, and runtime, while achieving extraction accuracy above 99.7\%. Building on this property, we further develop a two-channel covert communication mechanism in which ReTokSync serves as the primary channel and a reliable auxiliary channel corrects the remaining errors, achieving 100\% end-to-end recovery across all evaluated configurations.

\end{abstract}





\received{20 February 2007}
\received[revised]{12 March 2009}
\received[accepted]{5 June 2009}

\maketitle

\section{Introduction}

Recent advances in large language models (LLMs)~\cite{openai2023gpt4,anthropic2024claude3,geminiteam2023gemini,touvron2023llama2} have substantially improved the practical viability of generative linguistic steganography (GLS)~\cite{ziegler2019neural,zhang2021provably,kaptchuk2021meteor,ding2023discop,liao2025framework,bai2025shimmer,wang2025sparsamp,wang2026anstega}. Unlike traditional text steganography, which embeds information by modifying an existing carrier, GLS encodes secret messages directly into the autoregressive generation process of a language model. At each generation step, the steganographic algorithm selects a token from the model's next-token distribution such that the resulting text remains natural while implicitly conveying hidden information.

A central obstacle to the practical deployment of GLS is \emph{tokenization ambiguity}~\cite{nozaki2022addressing}. Modern language models commonly rely on subword tokenization schemes such as BPE, whose vocabularies contain many prefix-related units~\cite{sennrich2016rare,wu2016gnmt,kudo2018sentencepiece,kudo2018subword}. Consequently, the same surface string may admit multiple valid tokenization paths. In GLS, however, correct extraction requires the sender and the receiver to follow the same token sequence. Once the receiver re-tokenizes the stegotext differently from the sender, the two parties no longer share the same decoding state, and the resulting mismatch can propagate through subsequent probability reconstruction and message extraction, ultimately causing extraction failure~\cite{ueoka2021frustratingly}. Effective handling of tokenization ambiguity is therefore essential for reliable and deployable generative steganography~\cite{bauer2024leveraging}.

Existing studies on tokenization disambiguation in GLS can be broadly divided into two categories: \emph{token-removal methods} and \emph{distribution-preserving disambiguation methods}. Token-removal methods improve decodability by excluding risky candidates from the sender-side candidate set, but they explicitly alter the model's original next-token distribution and therefore cannot provide strict distributional security~\cite{yan2025addressingti,nozaki2022addressing,yan2023securedisambiguating}. Distribution-preserving methods instead seek to maintain the original generation distribution while ensuring reliable recovery. Pool-based methods achieve this goal by synchronizing over ambiguity-aware token groups, but at the cost of reduced embedding capacity~\cite{qi2025syncpool}. Verification-based methods are more flexible, yet their receiver-side checking incurs substantial runtime overhead in realistic settings~\cite{yan2024nearimperceptible}. Taken together, existing approaches reveal a persistent tension between security and efficiency.


We argue that this tension stems from two insufficiently exploited properties of tokenization ambiguity. First, ambiguity is typically triggered only at sparse local positions, yet existing methods impose global intervention at every step. Second, the main damage is not the local error itself but the subsequent sender--receiver desynchronization that propagates to all later extraction steps. As illustrated in Figure~\ref{fig:ambiguity-impact}, even a single local mismatch may suffice to invalidate all subsequent extraction results.

To address this problem, we propose ReTokSync, a self-synchro\-nizing
framework for tokenization disambiguation in GLS. ReTokSync tracks the
receiver-\allowbreak view tokenization result during generation, detects
ambiguity online, and immediately resets the sender's embedding state using
the corresponding decoding result. By doing so, it confines the effect of tokenization ambiguity to local residual deviations instead of allowing it to escalate into persistent desynchronization and global extraction failure. Because correction is triggered only when ambiguity actually occurs, ReTokSync leaves ambiguity-free positions unchanged and remains compatible with the underlying steganographic algorithm. Building on this property, we further develop a two-channel communication mechanism for continuous interaction, in which ReTokSync serves as the primary channel and Syncpool serves as a reliable auxiliary channel for correcting residual errors~\cite{qi2025syncpool}.

Experiments on both English and Chinese settings show that ReTokSync remains closest to the underlying steganographic baseline in distributional security, text quality, embedding capacity, and runtime, while reducing ambiguity-induced failures to sparse residual errors. Under the proposed two-channel communication mechanism, these residual errors can be fully repaired, enabling reliable end-to-end recovery.

The main contributions of this paper are summarized as follows:
\begin{enumerate}

    \item We identify a key limitation of existing disambiguation methods for GLS: tokenization ambiguity is triggered only at sparse local positions, yet existing methods impose global intervention at every generation step. We provide empirical evidence that ambiguity-triggering tokens consistently account for less than 1\% of the generated sequence, even as sample-level ambiguity frequency grows with sequence length.

    
    \item We propose \textbf{ReTokSync}, a self-synchronizing disambiguation framework that performs sender-side online ambiguity detection and corrective reset only when needed, thereby preventing local mismatch from propagating into global extraction failure. ReTokSync achieves zero KL divergence from the baseline distribution with less than 2\% relative time overhead, while reducing ambiguity-induced failures to sparse residual bit errors.


    \item We further develop a two-channel communication mechanism for continuous interaction, which combines ReTokSync as a high-efficiency primary channel with a low-rate auxiliary channel for residual-error correction, achieving 100\% end-to-end recovery across all evaluated configurations.

\end{enumerate}

\begin{figure}[t]
  \centering
  \includegraphics[width=\linewidth]{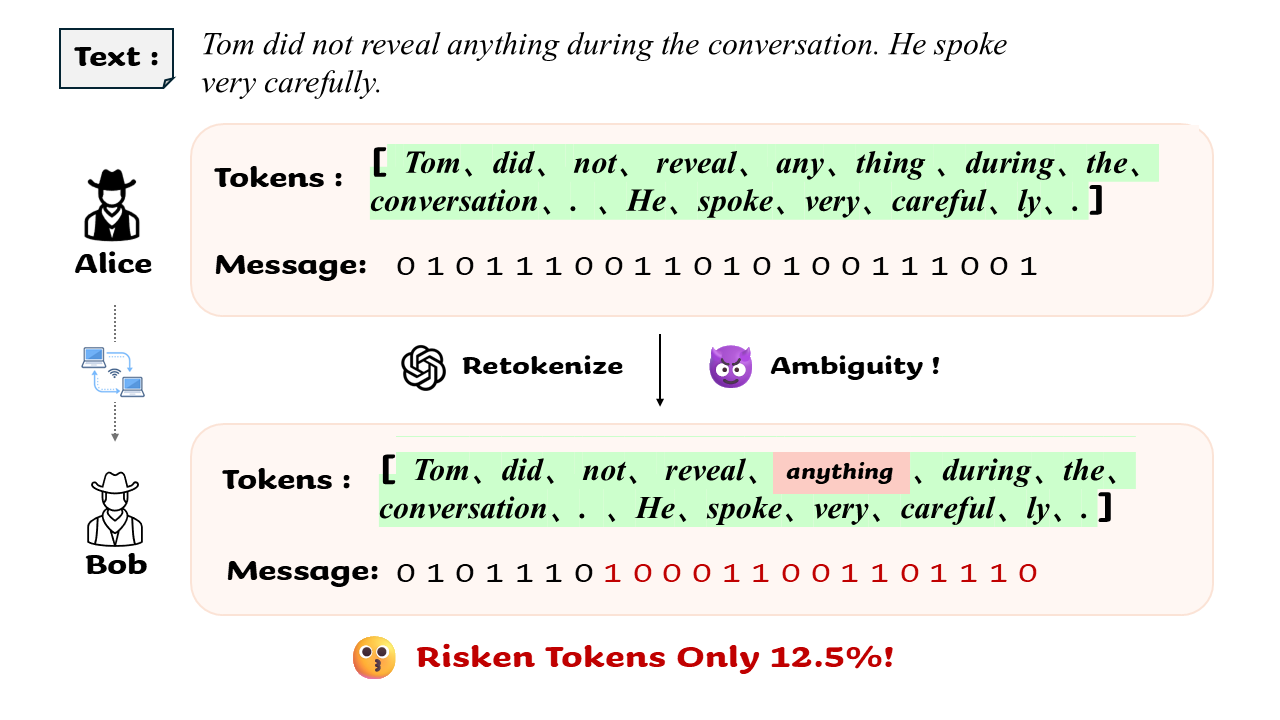}
  \caption{Impact of tokenization ambiguity in generative linguistic steganography. Although ambiguity is pervasive at the sample level, the tokens that trigger it are sparse. Once triggered, a single local mismatch can cause global extraction failure.}

  \Description{Two examples illustrating tokenization ambiguity in generative linguistic steganography. In both cases, only a small number of risky tokens trigger ambiguity, showing the sparsity of ambiguity-triggering tokens. However, once ambiguity occurs, the receiver tokenizes the same visible text differently from the sender, loses synchronization, and subsequently extracts an incorrect message.}
  \label{fig:ambiguity-impact}
\end{figure}

\section{Related Work}

Existing approaches to tokenization disambiguation in generative linguistic steganography can be broadly grouped into two categories: token-removal methods and distribution-preserving methods.

Token-removal methods resolve ambiguity by excluding candidate tokens that may induce inconsistent re-tokenization at the receiver. Nozaki and Murawaki~\cite{nozaki2022addressing} remove prefix-conflicting tokens to prevent multiple tokenization outcomes for the same surface string. Yan et al.~\cite{yan2023securedisambiguating} further propose MWIS, which formulates disambiguation as a maximum weight independent set problem so as to retain as much probability mass as possible under token-conflict constraints. Yan and Murawaki subsequently introduce Stepwise Verification~\cite{yan2025addressingti,yan2025lowoverhead}, which identifies risky candidates through stepwise re-encoding and disables tokens that may cause ambiguity at the current generation step. Although these methods improve decodability, they all rely on modifying the sender-side candidate set before sampling. As a consequence, they cannot preserve the original next-token distribution exactly. From the perspective of distributional security, this means that ambiguity is avoided by introducing distribution distortion, which can be reflected in non-zero KL divergence from the underlying model distribution. Therefore, token-removal methods are difficult to reconcile with the strict security requirement of provably secure steganographic generation.

Distribution-preserving methods instead seek to maintain the original generation distribution while still enabling reliable recovery. Qi et al.~\cite{qi2025syncpool} propose SyncPool, which constructs ambiguity-aware token pools and uses these pools, rather than individual tokens, as the embedding unit. This design preserves distributional security, but sacrifices embedding capacity because part of the available entropy must be reserved for synchronization within the pool. Yan et al.~\cite{yan2024nearimperceptible} propose a verification-based method that shifts disambiguation to the receiver by enumerating possible re-tokenization paths and verifying which path matches the sender's actual token sequence. This strategy avoids direct distortion of the sender-side distribution, but incurs substantial decoding overhead, especially when the candidate space becomes large. Thus, even among distribution-preserving methods, the cost of disambiguation remains significant: one line of work pays with reduced embedding efficiency, whereas the other pays with prohibitive runtime overhead.


Overall, existing methods reveal that the central difficulty lies not in whether ambiguity can be removed, but in how to do so without distorting the generation distribution or incurring excessive overhead. Motivated by the observation that ambiguity is globally possible but locally sparse, and that its primary damage stems from desynchronization rather than the local error itself, we design a self-synchronizing framework that intervenes only when ambiguity actually occurs.

\section{Method}

\subsection{The ReTokSync Framework}

Guided by the above observation, we present \textbf{ReTokSync}, a self-synchronizing framework for tokenization disambiguation in GLS. Unlike existing approaches that either impose global constraints on the generation process or rely on receiver-side verification over multiple tokenization paths, ReTokSync addresses ambiguity through targeted synchronization maintenance. During generation, it continuously tracks the token sequence that the receiver would obtain by re-tokenizing the current surface text. Whenever this receiver-view sequence diverges from the sender's predicted sequence, ReTokSync performs a corrective reset based on the corresponding receiver-view decoding result, so that subsequent embedding can proceed from a synchronized state.

This design directly targets the practical consequence of tokenization ambiguity, namely error propagation caused by sender--receiver desynchronization. Because intervention is triggered only when ambiguity is actually introduced, ReTokSync leaves ambiguity-free generation steps unchanged and avoids unnecessary interference with the underlying steganographic process. As a result, the framework remains compatible with the base steganographic scheme and is designed to preserve its security and efficiency as much as possible while preventing local mismatch from escalating into global extraction failure. The complete procedure is given in Algorithm~\ref{alg:retoksync}.


\begin{figure*}[t]
  \centering
  \includegraphics[width=\textwidth]{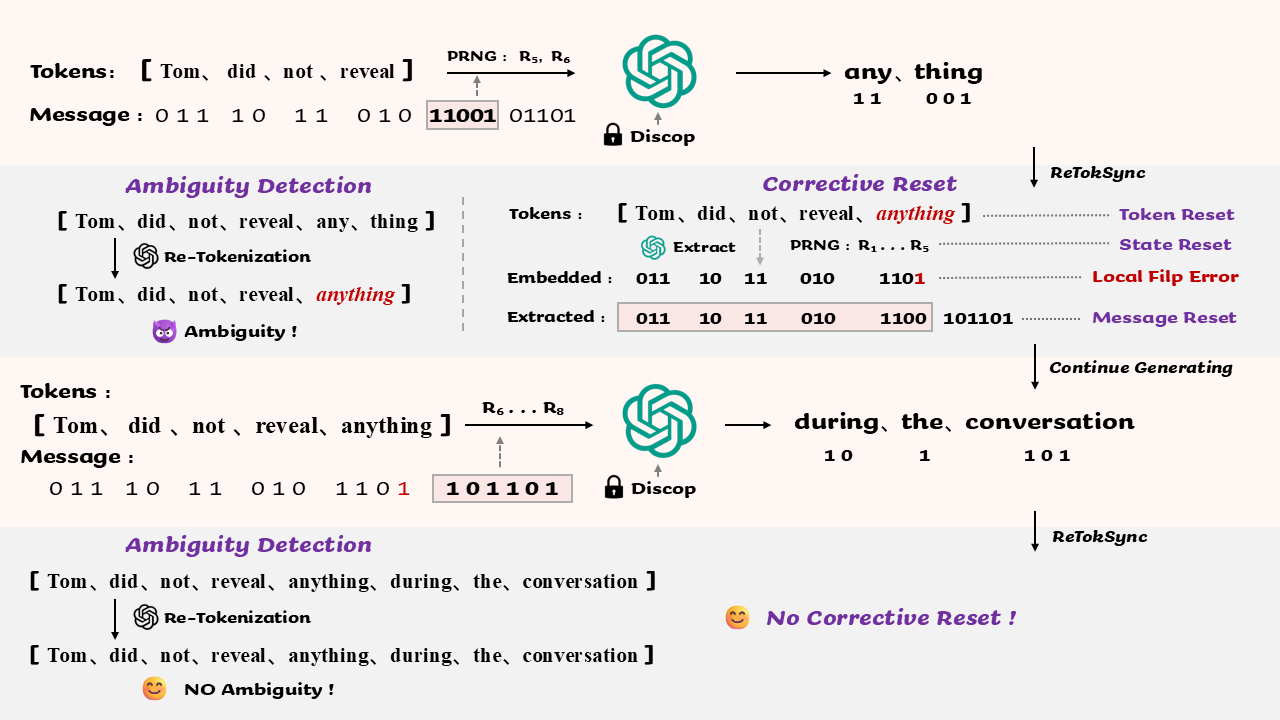}
  \caption{Illustration of ReTokSync. When tokenization ambiguity is detected, the sender performs corrective reset to prevent ambiguity-induced errors from propagating to subsequent positions; when no ambiguity occurs, no corrective reset is needed.}
  \Description{Illustration of ReTokSync. Upon detecting tokenization ambiguity, the sender performs corrective reset to prevent ambiguity-induced errors from propagating to subsequent positions. When no ambiguity occurs, no corrective reset is required.}
  \label{fig:retoksync-example}
\end{figure*}

\subsection{Re-tokenization}

During generation, the sender maintains two token sequences: the \emph{true sequence} $x$ and the \emph{re-tokenized sequence} $\tilde{x}$. The true sequence $x$ records the tokens actually generated by the sender, while the re-tokenized sequence $\tilde{x}$ represents the tokenization result that would be obtained by the receiver from the visible text. Formally, after appending a newly generated token $x_t$ to the current true sequence, the sender updates
\begin{equation}
\tilde{x} = \operatorname{Tok}(\operatorname{Detok}(x)).
\label{eq:retok}
\end{equation}

Crucially, subsequent generation is conditioned on $\tilde{x}$ rather than on $x$. This design ensures that the sender always uses the same token context as the receiver, because $\tilde{x}$ is exactly the sequence obtained by re-tokenizing the visible text decoded from $x$. As a result, both parties rely on the same conditional distribution at each step:
\begin{equation}
p_{\theta}(\cdot \mid \tilde{x})
=
p_{\theta}(\cdot \mid x^{(R)}),
\label{eq:shared-dist}
\end{equation}
where $x^{(R)} = \operatorname{Tok}(\operatorname{Detok}(x))$ denotes the receiver-side tokenization result for the same visible text. Sharing the same conditional distribution is the prerequisite for maintaining correct synchronization under tokenization ambiguity.

\subsection{Ambiguity Detection}

At each generation step, let $x_t$ denote the newly sampled token. Before re-tokenization, the sender constructs the predicted receiver-view sequence
\begin{equation}
\tilde{x}^{\mathrm{pred}} = \tilde{x} \mathbin{\|} [x_t],
\label{eq:pred}
\end{equation}
where $\|$ denotes concatenation. After appending $x_t$ to the true sequence, the sender re-tokenizes the updated text and obtains the actual receiver-view sequence
\begin{equation}
\tilde{x}^{\mathrm{retok}} = \operatorname{Tok}(\operatorname{Detok}(x)).
\label{eq:retok-actual}
\end{equation}

Tokenization ambiguity is detected by the consistency check
\begin{equation}
\tilde{x}^{\mathrm{pred}} \stackrel{?}{=} \tilde{x}^{\mathrm{retok}}.
\label{eq:ambiguity-check}
\end{equation}
If Eq.~(\ref{eq:ambiguity-check}) holds, the new token does not alter the receiver's tokenization structure beyond direct appending. Otherwise, ambiguity has been introduced, and the sender must perform a corrective reset. This criterion directly matches the operational test used in Algorithm~\ref{alg:retoksync}.

\subsection{Corrective Reset}

When ambiguity is detected, the sender decodes the current receiver-view sequence $\tilde{x}^{\mathrm{retok}}$ and uses the decoding result to reset its embedding state:
\begin{equation}
(\hat{m}_{1:\hat{j}}, \hat{s}) \leftarrow \mathsf{Dec}(\tilde{x}^{\mathrm{retok}}, s_0),
\label{eq:decode-reset}
\end{equation}
followed by
\begin{equation}
j \leftarrow \hat{j}, \qquad s \leftarrow \hat{s}.
\label{eq:state-reset}
\end{equation}
Here, $\hat{j}$ denotes the message pointer recovered from the receiver-view sequence, and $\hat{s}$ denotes the corresponding internal decoding state. The sender then continues generation from the reset state while replacing $\tilde{x}$ with $\tilde{x}^{\mathrm{retok}}$.

As illustrated in Figure~\ref{fig:retoksync-example}, when the tokens \textit{any} and \textit{thing} are re-tokenized as \textit{anything}, ambiguity is detected and ReTokSync resets the token sequence, PRNG state, and embedding pointer according to the receiver-view result. Because message extraction is then performed on \textit{anything}, this step may introduce a local 1-bit error. In the subsequent generation, no further ambiguity is detected, and thus no additional reset is required. 

\paragraph{Security analysis.}
Two properties underpin ReTokSync's distributional security.
First, because generation is always conditioned on $\tilde{x}$ (Eq.~\ref{eq:shared-dist}), the conditional distribution at every step is identical to the one the receiver observes; this holds regardless of whether a reset has just occurred. Second, the corrective reset operates exclusively on the \emph{internal} embedding state ($j$, $s$) and does not modify the next-token distribution seen by an external observer. Consequently, the token-level generation distribution under ReTokSync coincides with that of the underlying steganographic scheme, yielding zero KL divergence. A formal analysis is provided in the supplementary material.

\begin{algorithm}[t]
\caption{ReTokSync: Correction-and-Reset under Tokenization Ambiguity}
\label{alg:retoksync}
\begin{algorithmic}[1]
\STATE \textbf{Require:} payload bitstring $\mathbf{m}$, tokenizer $\operatorname{Tok}/\operatorname{Detok}$, stego encoder $\mathsf{Enc}$, stego decoder $\mathsf{Dec}$, initial state $s_0$
\STATE $x \gets [\,]$, $\tilde{x} \gets [\,]$, $j \gets 0$, $s \gets s_0$
\FOR{$t = 1,2,\ldots$}
    \STATE $(x_t, j, s) \gets \mathsf{Enc}(\tilde{x}, \mathbf{m}, j, s)$
    \STATE $x \gets x \mathbin{\|} [x_t]$
    \STATE $\tilde{x}^{\mathrm{pred}} \gets \tilde{x} \mathbin{\|} [x_t]$
    \STATE $\tilde{x}^{\mathrm{retok}} \gets \operatorname{Tok}(\operatorname{Detok}(x))$
    \IF{$\tilde{x}^{\mathrm{pred}} \neq \tilde{x}^{\mathrm{retok}}$}
        \STATE $(\hat{m}_{1:\hat{j}}, \hat{s}) \gets \mathsf{Dec}(\tilde{x}^{\mathrm{retok}}, s_0)$
        \STATE $j \gets \hat{j}$, $s \gets \hat{s}$
    \ENDIF
    \STATE $\tilde{x} \gets \tilde{x}^{\mathrm{retok}}$
\ENDFOR
\end{algorithmic}
\end{algorithm}

\begin{figure}[t]
  \centering
  \includegraphics[width=\columnwidth]{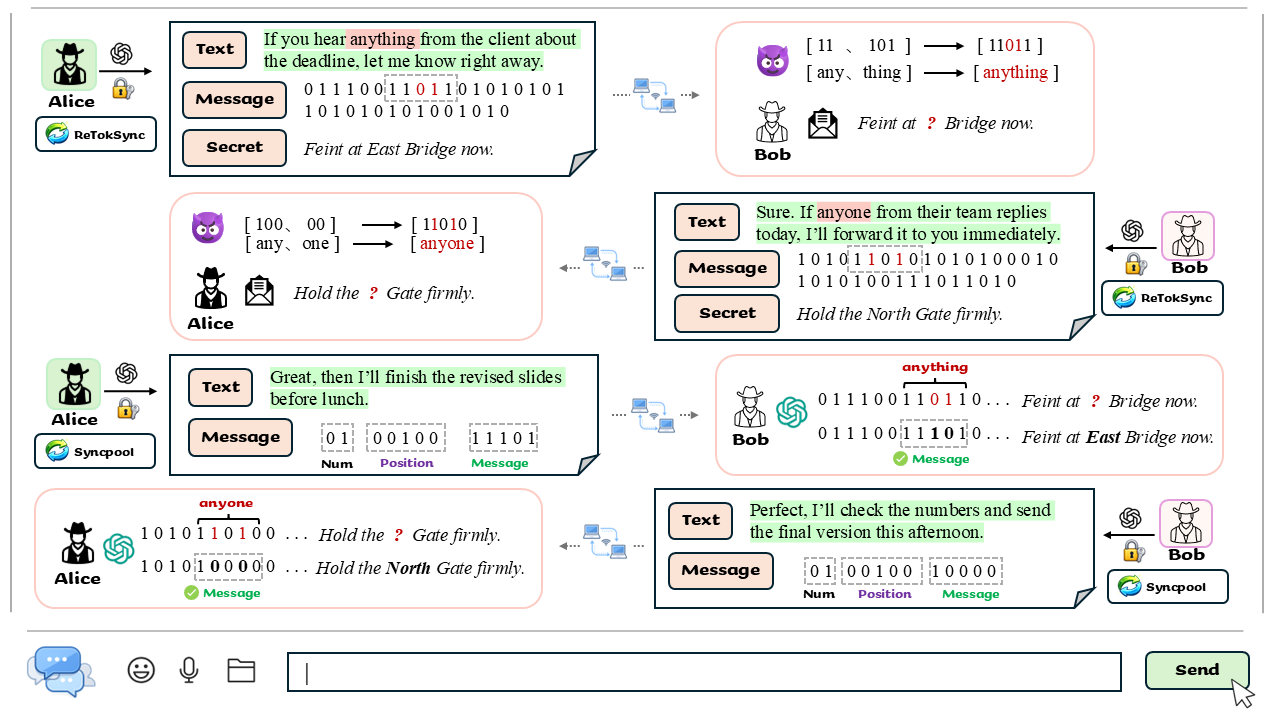}
  \caption{An example of covert communication under continuous interaction. ReTokSync is used as the main channel to transmit the primary secret payload, while Syncpool serves as an auxiliary reliable channel that periodically carries correction information in later interaction rounds, enabling accurate end-to-end recovery.}
  \Description{An example of multi-turn covert communication in a continuous interaction setting. ReTokSync is used to transmit the main secret payload, while Syncpool carries correction information in later interaction rounds to repair residual local errors and enable accurate end-to-end recovery.}
  \label{fig:continuous-communication}
\end{figure}

\subsection{A Covert Communication Mechanism for Continuous Interaction}


Building on the ReTokSync framework, we develop a covert communication mechanism that extends single-sample disambiguation to practical continuous interaction, enabling system-level handling of tokenization ambiguity.

\subsubsection{Background and Design Basis}


Existing research on generative steganography typically models communication as a one-shot process of text generation and message extraction \cite{ding2023discop}. In realistic public channels, however, covert information transmission more often takes place through continuous language interaction---email exchanges, instant messaging, or multi-turn online conversations. Covert communication may likewise arise in multi-agent interaction \cite{motwani2024secretcollusion,liu2024iagents}. To handle tokenization ambiguity reliably without sacrificing security or efficiency, we construct a two-channel communication mechanism: ReTokSync serves as the primary channel for transmitting the main secret payload, while Syncpool serves as an auxiliary reliable channel for correction \cite{qi2025syncpool}.


This design rests on three observations. First, ReTokSync confines ambiguity-induced damage to a very small number of bit flips rather than global desynchronization, substantially reducing the correction payload. Second, because the sender maintains synchronization with the receiver, it can predict the receiver's extraction result at each token, making structured correction feasible. Third, ReTokSync and Syncpool are complementary: the former provides high embedding efficiency with sparse residual errors, while the latter guarantees error-free disambiguation at lower embedding capacity. Together, they enable efficient primary transmission with lightweight residual-error repair.

\subsubsection{Workflow}

In operation, we adopt a group-wise correction strategy. After transmitting \(n\) steganographic samples through the primary channel, the sender sends one additional correction message for that group through the auxiliary channel. This message provides a unified description of all residual errors within the group.

On the sender side, the receiver-view token sequences of all samples in the group are first concatenated in order to form a group-level token sequence. For each token, the sender records the message fragment that would actually be extracted by the receiver during decoding. Since the sender also knows the original target payload, it can align the extracted message stream with the intended one at the token level, thereby identifying the tokens that contain errors and the correct message fragments that should replace them. The correction message is composed of a set of correction items, each corresponding to a token carrying erroneous bits and specifying the correct message fragment used to replace its erroneous extraction result.

As illustrated in Figure~\ref{fig:continuous-communication}, each correction message contains the number of corrections, the token positions to be corrected, and the corresponding replacement message fragments. The detailed message format and its construction are described in the supplementary material.

On the receiver side, standard extraction is first performed on all samples in the group, and the corresponding token-level message fragments are recorded in the same order. The receiver then parses the correction message transmitted through the auxiliary channel, locates the tokens to be corrected, and replaces the originally extracted erroneous fragments with the corrected ones specified in the message. After all corrections are applied, the receiver obtains the correct group-level message stream, which can then be mapped back to the extraction results of individual samples, thereby recovering the final secret message for each sample.

\section{Experiments}

\begin{figure}[t]
  \centering
  \includegraphics[width=\linewidth]{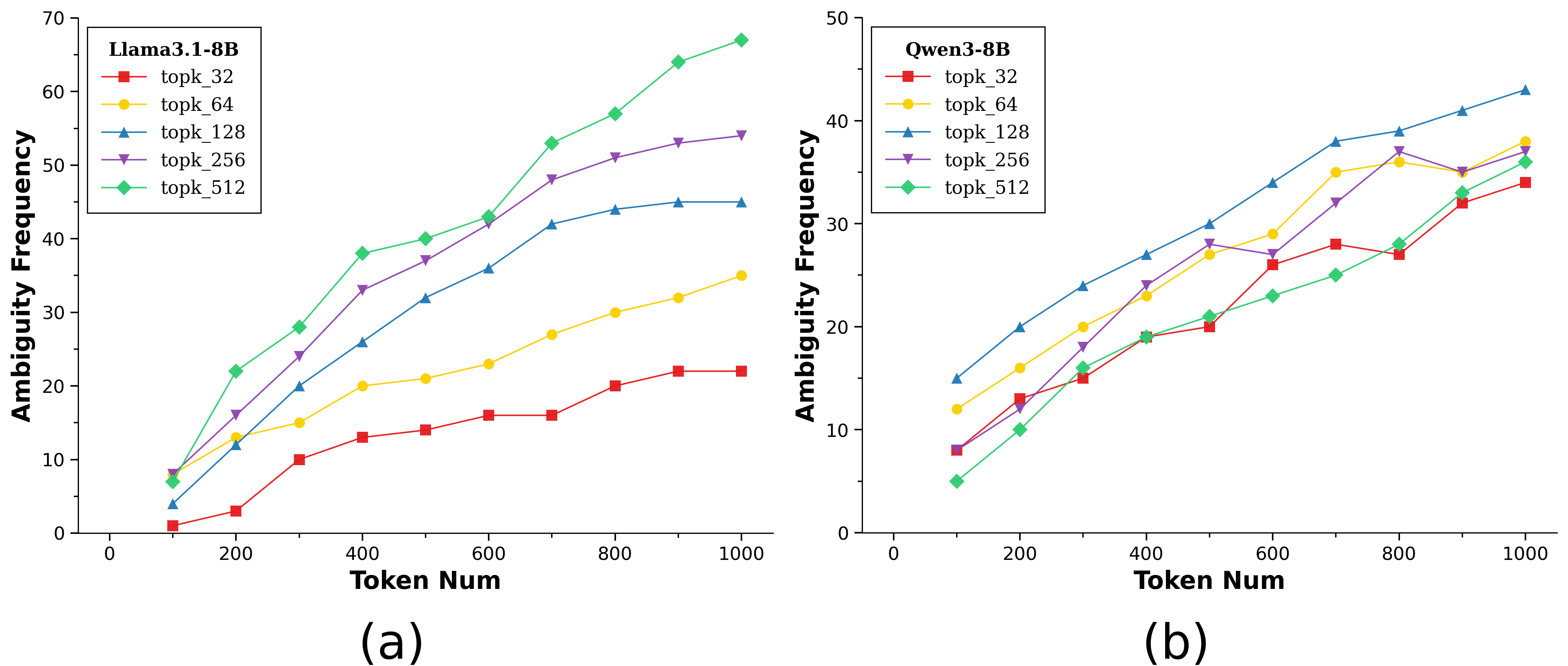}
  \caption{Tokenization-ambiguity frequency analysis on IMDB using Llama-3.1-8B and Qwen3-8B.}
  \Description{A figure showing how often tokenization ambiguity occurs under different generation lengths and top-k settings on the IMDB dataset for Llama-3.1-8B and Qwen3-8B.}
  \label{fig:frequency1}
\end{figure}

\begin{figure}[t]
  \centering
  \includegraphics[width=\linewidth]{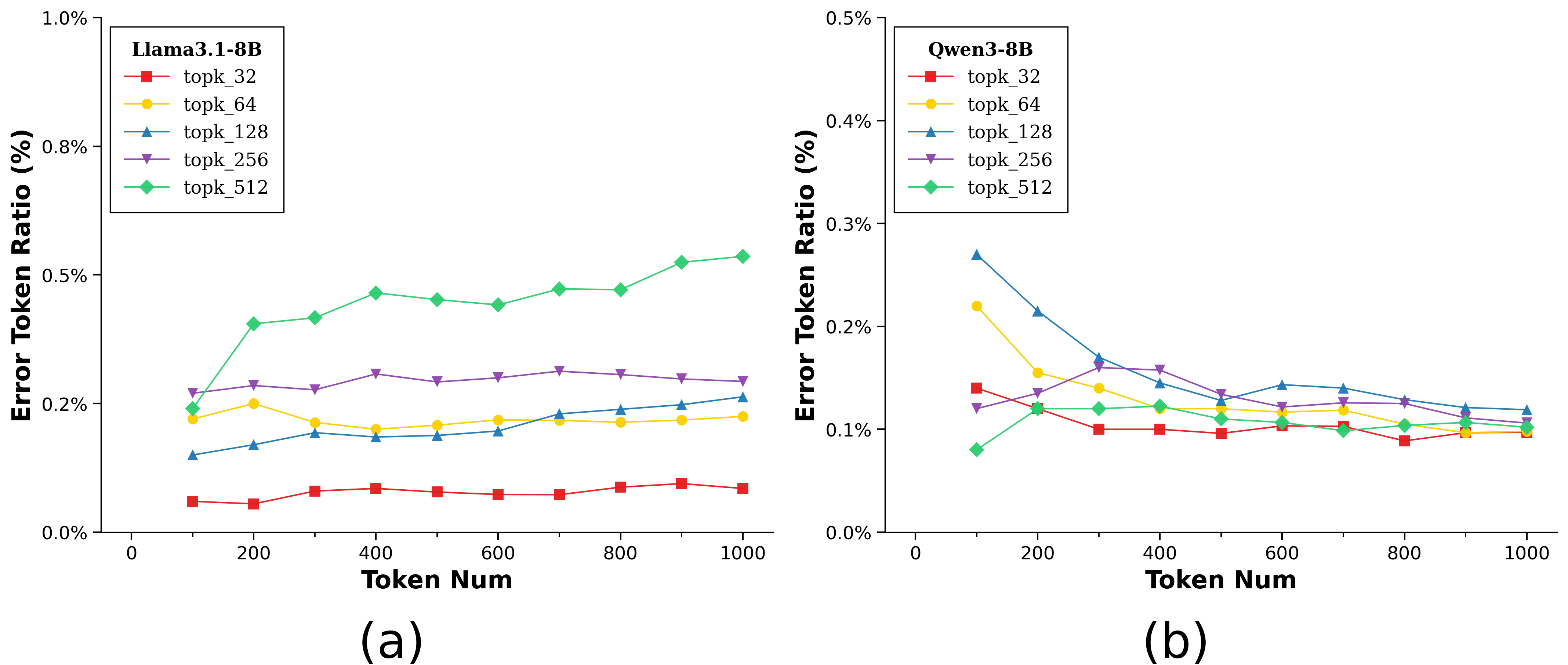}
  \caption{Frequency analysis of ambiguity-triggering tokens on IMDB using Llama-3.1-8B and Qwen3-8B.}
  \Description{A figure showing the proportion of generated tokens that actually trigger tokenization ambiguity under different generation lengths and top-k settings on the IMDB dataset for Llama-3.1-8B and Qwen3-8B.}
  \label{fig:frequency2}
\end{figure}

\subsection{Experimental Setup}

We primarily evaluate our method on the Llama-3.1-8B model \cite{grattafiori2024llama3herd} using the IMDB English dataset \cite{maas2011learning}.
To further assess the generality of the proposed disambiguation strategy, we also conduct experiments with Qwen3-8B \cite{yang2025qwen3} on a Chinese-translated version of the IMDB dataset.
For the steganographic task, the secret message is a randomly generated binary bitstream.
In each run, we use the first three sentences of a dataset example as the initial context and generate 100 subsequent tokens.

We adopt top-$k$ truncation \cite{holtzman2020curious} during generation, using $k=32$, $128$, and $512$ to directly examine how candidate-pool size affects disambiguation behavior and overall performance.


We adopt Discop \cite{ding2023discop} as the main baseline, as it is a secure and efficient steganographic scheme that preserves the underlying generation distribution. Note that Discop does not perform any disambiguation; its extraction accuracy is marked ``--'' in Tables~\ref{tab:main-english}--\ref{tab:main-chinese} because, without disambiguation, tokenization ambiguity causes the receiver to lose synchronization and subsequent extraction becomes unreliable (see Section~\ref{sec:ambiguity-frequency} for the high sample-level ambiguity frequency). Beyond our proposed ReTokSync, we further compare with three token-removal-based disambiguation methods, including Basic \cite{nozaki2022addressing}, MWIS \cite{yan2023securedisambiguating}, and Stepwise Verification \cite{yan2025addressingti,yan2025lowoverhead}, as well as two distribution-preserving disambiguation strategies, including Syncpool \cite{qi2025syncpool}, and Verification-based \cite{yan2024nearimperceptible}. This setup enables a comprehensive evaluation across different classes of ambiguity-handling methods. For completeness, we also evaluate ReTokSync on additional steganographic algorithms; detailed results are provided in the supplementary material.

All experiments are implemented in PyTorch~2.5 with CUDA~11.8 and are run on a machine equipped with an Intel Xeon Gold~6330 CPU and an NVIDIA RTX~4090 GPU.

\subsection{Evaluation Metrics}
\label{sec:metrics}
To evaluate disambiguation strategies in generative steganography, we assess
three dimensions: \emph{effectiveness}, \emph{security}, and \emph{efficiency}.

\paragraph{Effectiveness.}
\textbf{Extraction accuracy} measures whether the receiver can correctly recover
the embedded secret message. We report bit-wise accuracy between the extracted
bitstream and the ground-truth message.

\paragraph{Security.}
\textbf{Kullback--Leibler divergence (KLD)} measures how much the token
distribution under a disambiguation method deviates from the model's original
generation distribution. Lower KLD indicates less perturbation to the original
distribution and therefore stronger distributional security.

\textbf{Average KLD} and \textbf{Maximum KLD} capture such perturbation from the
perspectives of average-case and worst-case distortion, respectively.

\textbf{Perplexity (PPL)} evaluates the fluency and naturalness of the
generated text:
\begin{equation}
\mathrm{PPL}
=
2^{-\frac{1}{N}\sum_{i=1}^{N}\log_2 \Pr(x_i \mid x_{<i}) }.
\end{equation}
A PPL closer to that of the original stego generation indicates stronger
security.

\paragraph{Efficiency.}
\textbf{Capacity} measures the average number of embedded bits per generated
token.

\textbf{Entropy utilization} measures how effectively the available embedding
entropy is converted into payload bits.

\textbf{Time efficiency} reports the computational overhead of embedding and
extraction. We measure the actual wall-clock time of the complete embedding and
extraction process. In addition, we introduce \textbf{Relative Time Overhead
(RTO)} to characterize the time cost of each disambiguation method relative to
its corresponding baseline:
\begin{equation}
\mathrm{RTO} = \frac{T_{\text{method}} - T_{\text{baseline}}}{T_{\text{baseline}}} \times 100\%,
\end{equation}
where $T_{\text{method}}$ and $T_{\text{baseline}}$ denote the total runtime of
the evaluated method and its corresponding baseline, respectively.

\subsection{Ambiguity Frequency}
\label{sec:ambiguity-frequency}

To investigate the frequency of tokenization ambiguity, we conduct experiments
in both Chinese and English settings. Specifically, we use Llama-3.1-8B and
Qwen3-8B to evaluate on the English and Chinese versions of the IMDB dataset,
respectively. For each configuration of generation length and top-$k$, we test
100 samples.

We report two statistics: the proportion of samples in which ambiguity occurs,
and the proportion of tokens that trigger ambiguity. The sample-level ratio
reflects how frequently tokenization ambiguity may arise when the receiver
re-encodes the received text, whereas the token-level ratio measures the
fraction of tokens that actually induce ambiguity among all generated tokens.

Fig.~\ref{fig:frequency1} and Fig.~\ref{fig:frequency2} show that tokenization
ambiguity is a pervasive phenomenon, and its frequency consistently increases
with the generation length. In contrast, the proportion of ambiguity-triggering
tokens remains consistently low. This indicates that, although tokenization
ambiguity occurs frequently at the sample level, its triggering positions are
distributed sparsely in the token sequence. This observation further suggests
that existing strategies that continuously intervene at every generation
position are inherently inefficient, since most positions do not actually lead
to tokenization ambiguity.


Motivated by this observation, ReTokSync adopts a post-processing perspective:
it detects ambiguity in real time and performs correction only when necessary.
Because ambiguity-triggering positions account for only a small fraction of the
sequence, ReTokSync requires only lightweight detection at the vast majority of
non-ambiguous positions, avoiding unnecessary overhead while preserving normal
generation behavior.

\begin{table*}[t]
\centering
\caption{Comparison of disambiguation strategies on Llama-3.1-8B across different top-$k$ settings on the English IMDB dataset.}
\label{tab:main-english}
\begingroup
\setlength{\tabcolsep}{5pt}
\renewcommand{\arraystretch}{1.08}
\resizebox{\textwidth}{!}{%
\begin{tabular}{c|c|c S S S S S S c S}
\toprule
\multicolumn{2}{c|}{Method} & {$k$} &
\multicolumn{1}{c}{Ave PPL} &
\multicolumn{1}{c}{Ave KLD} &
\multicolumn{1}{c}{Max KLD} &
\multicolumn{1}{c}{Capacity} &
\multicolumn{1}{c}{Utilization} &
\multicolumn{1}{c}{Total Time} &
\multicolumn{1}{c}{RTO} &
\multicolumn{1}{c}{Accuracy} \\
\multicolumn{2}{c|}{} & {} &
\multicolumn{1}{c}{} &
\multicolumn{1}{c}{} &
\multicolumn{1}{c}{} &
\multicolumn{1}{c}{(bits/token)} &
\multicolumn{1}{c}{} &
\multicolumn{1}{c}{(s)} &
\multicolumn{1}{c}{(\%)} &
\multicolumn{1}{c}{} \\
\midrule

\multicolumn{2}{c|}{\multirow[c]{3}{*}{Discop~\cite{ding2023discop}}} & 32  & 7.26 & 0.00 & 0.00 & 2.28 & 0.88 & 4.90 & 0.00 & \multicolumn{1}{c}{--} \\
\multicolumn{2}{c|}{} & 128 & 11.33 & 0.00 & 0.00 & 2.95 & 0.91 & 4.98 & 0.00 & \multicolumn{1}{c}{--} \\
\multicolumn{2}{c|}{} & 512 & 17.29 & 0.00 & 0.00 & 3.50 & 0.91 & 5.18 & 0.00 & \multicolumn{1}{c}{--} \\
\midrule

\multirow{9}{*}{\parbox[c]{1.8cm}{\centering Token Removal}} &
\multirow{3}{*}{Basic~\cite{nozaki2022addressing}}
  & 32  & 19.93 & 0.96 & 9.15 & 2.72 & 0.96 & 7.64 & 55.92 & 1.00 \\
& & 128 & 59.05 & 1.24 & 8.77 & 4.05 & 0.98 & 8.75 & 75.70 & 1.00 \\
& & 512 & 219.67 & 1.36 & 7.58 & 5.72 & 0.98 & 21.69 & 318.73 & 1.00 \\
\cmidrule(lr){2-11}

& \multirow{3}{*}{MWIS~\cite{yan2023securedisambiguating}}
  & 32  & 6.87 & 0.09 & 0.75 & 2.12 & 0.81 & 7.64 & 55.92 & 1.00 \\
& & 128 & 9.34 & 0.11 & 0.75 & 2.61 & 0.82 & 9.01 & 80.92 & 1.00 \\
& & 512 & 11.74 & 0.15 & 0.83 & 2.93 & 0.80 & 23.89 & 361.20 & 1.00 \\
\cmidrule(lr){2-11}

& \multirow{3}{*}{\parbox[c]{2.9cm}{\centering Stepwise\\Verification~\cite{yan2025addressingti,yan2025lowoverhead}}}
  & 32  & 7.36 & 0.00 & 0.05 & 2.31 & 0.89 & 6.65 & 35.71 & 1.00 \\
& & 128 & 11.93 & 0.00 & 0.15 & 3.06 & 0.91 & 11.82 & 137.35 & 1.00 \\
& & 512 & 16.37 & 0.00 & 0.01 & 3.45 & 0.91 & 32.39 & 525.29 & 1.00 \\

\midrule
\multirow{6}{*}{\parbox[c]{2.4cm}{\centering Distribution\\Preserving}} &
\multirow{3}{*}{Syncpool~\cite{qi2025syncpool}}
  & 32  & 7.65 & 0.00 & 0.00 & 1.68 & 0.64 & 5.05 & 3.06 & 1.00 \\
& & 128 & 11.98 & 0.00 & 0.00 & 1.01 & 0.33 & 5.15 & 3.41 & 1.00 \\
& & 512 & 17.06 & 0.00 & 0.00 & 0.53 & 0.14 & 5.21 & 0.58 & 1.00 \\
\cmidrule(lr){2-11}

& \multirow{3}{*}{\parbox[c]{2.8cm}{\centering Verification\\-based~\cite{yan2024nearimperceptible}}}
  & 32  & 6.58 & 0.00 & 0.00 & 2.32 & 0.86 & 15.17 & 209.59 & 1.00 \\
& & 128 & 11.00 & 0.00 & 0.00 & 2.82 & 0.86 & 441.21 & 8759.64 & 1.00 \\
& & 256 & 13.54 & 0.00 & 0.00 & 3.13 & 0.87 & 552.64 & 10714.87 & 1.00 \\

\midrule
\multirow{3}{*}{Ours} & \multirow{3}{*}{ReTokSync}
  & 32  & 7.26 & 0.00 & 0.00 & 2.28 & 0.88 & 4.94 & 0.82 & 1.00 \\
& & 128 & 11.36 & 0.00 & 0.00 & 2.95 & 0.91 & 5.05 & 1.41 & 0.999 \\
& & 512 & 17.29 & 0.00 & 0.00 & 3.50 & 0.91 & 5.27 & 1.74 & 1.00 \\

\bottomrule
\end{tabular}%
}
\endgroup
\end{table*}

\begin{table*}[t]
\centering
\caption{Comparison of disambiguation strategies on Qwen3-8B across different top-$k$ settings on the Chinese IMDB dataset.}
\label{tab:main-chinese}
\begingroup
\setlength{\tabcolsep}{5pt}
\renewcommand{\arraystretch}{1.08}
\resizebox{\textwidth}{!}{%
\begin{tabular}{c|c|c S S S S S S c S}
\toprule
\multicolumn{2}{c|}{Method} & {$k$} &
\multicolumn{1}{c}{Ave PPL} &
\multicolumn{1}{c}{Ave KLD} &
\multicolumn{1}{c}{Max KLD} &
\multicolumn{1}{c}{Capacity} &
\multicolumn{1}{c}{Utilization} &
\multicolumn{1}{c}{Total Time} &
\multicolumn{1}{c}{RTO} &
\multicolumn{1}{c}{Accuracy} \\
\multicolumn{2}{c|}{} & {} &
\multicolumn{1}{c}{} &
\multicolumn{1}{c}{} &
\multicolumn{1}{c}{} &
\multicolumn{1}{c}{(bits/token)} &
\multicolumn{1}{c}{} &
\multicolumn{1}{c}{(s)} &
\multicolumn{1}{c}{(\%)} &
\multicolumn{1}{c}{} \\
\midrule

\multicolumn{2}{c|}{\multirow[c]{3}{*}{Discop~\cite{ding2023discop}}} & 32  & 5.58 & 0.00 & 0.00 & 1.85 & 0.87 & 6.85 & 0.00 & \multicolumn{1}{c}{--} \\
\multicolumn{2}{c|}{} & 128 & 8.70 & 0.00 & 0.00 & 2.47 & 0.90 & 7.00 & 0.00 & \multicolumn{1}{c}{--} \\
\multicolumn{2}{c|}{} & 512 & 11.61 & 0.00 & 0.00 & 2.77 & 0.91 & 7.07 & 0.00 & \multicolumn{1}{c}{--} \\
\midrule

\multirow{9}{*}{\parbox[c]{1.8cm}{\centering Token Removal}} &
\multirow{3}{*}{Basic~\cite{nozaki2022addressing}}
  & 32  & 21.10 & 1.64 & 15.32 & 2.22 & 0.96 & 11.92 & 74.01 & 1.00 \\
& & 128 & 56.47 & 2.18 & 15.46 & 3.07 & 0.98 & 23.22 & 231.71 & 1.00 \\
& & 512 & 104.79 & 2.29 & 16.82 & 3.62 & 0.99 & 27.01 & 282.04 & 1.00 \\
\cmidrule(lr){2-11}

& \multirow{3}{*}{MWIS~\cite{yan2023securedisambiguating}}
  & 32  & 5.60 & 0.16 & 1.15 & 1.69 & 0.74 & 10.61 & 54.89 & 1.00 \\
& & 128 & 7.69 & 0.18 & 1.16 & 2.21 & 0.77 & 23.73 & 239.00 & 1.00 \\
& & 512 & 9.93 & 0.19 & 1.16 & 2.53 & 0.78 & 29.28 & 314.14 & 1.00 \\
\cmidrule(lr){2-11}

& \multirow{3}{*}{\parbox[c]{2.9cm}{\centering Stepwise\\Verification~\cite{yan2025addressingti,yan2025lowoverhead}}}
  & 32  & 5.54 & 0.04 & 1.28 & 1.87 & 0.83 & 10.10 & 47.45 & 1.00 \\
& & 128 & 16.76 & 0.04 & 1.48 & 2.34 & 0.85 & 19.19 & 174.14 & 1.00 \\
& & 512 & 21.64 & 0.04 & 1.39 & 2.76 & 0.86 & 55.51 & 685.15 & 1.00 \\

\midrule
\multirow{6}{*}{\parbox[c]{2.4cm}{\centering Distribution\\Preserving}} &
\multirow{3}{*}{Syncpool~\cite{qi2025syncpool}}
  & 32  & 6.43 & 0.00 & 0.00 & 1.60 & 0.69 & 6.93 & 1.17 & 1.00 \\
& & 128 & 8.56 & 0.00 & 0.00 & 1.95 & 0.69 & 8.49 & 21.29 & 1.00 \\
& & 512 & 11.71 & 0.00 & 0.00 & 2.24 & 0.70 & 15.83 & 123.90 & 1.00 \\
\cmidrule(lr){2-11}

& \multirow{3}{*}{\parbox[c]{2.8cm}{\centering Verification\\-based~\cite{yan2024nearimperceptible}}}
  & 32  & 5.78 & 0.00 & 0.00 & 1.61 & 0.83 & 45.51 & 564.38 & 1.00 \\
& & 128 & 9.17 & 0.00 & 0.00 & 2.32 & 0.87 & 400.13 & 5616.14 & 1.00 \\
& & 256 & 10.46 & 0.00 & 0.00 & 2.54 & 0.86 & 741.34 & 10770.09 & 1.00 \\

\midrule
\multirow{3}{*}{Ours} & \multirow{3}{*}{ReTokSync}
  & 32  & 5.53 & 0.00 & 0.00 & 1.85 & 0.87 & 8.16 & 19.12 & 0.998 \\
& & 128 & 8.70 & 0.00 & 0.00 & 2.47 & 0.90 & 7.40 & 5.71 & 0.998 \\
& & 512 & 11.00 & 0.00 & 0.00 & 2.77 & 0.91 & 8.07 & 14.14 & 0.997 \\

\bottomrule
\end{tabular}%
}
\endgroup
\end{table*}


\subsection{Results and Analysis}

We compare two mainstream categories of disambiguation strategies, together with ReTokSync, from three perspectives: \emph{effectiveness}, \emph{security}, and \emph{efficiency}. Results for the English and Chinese settings are reported in Table~\ref{tab:main-english} and Table~\ref{tab:main-chinese}, respectively. Here, \emph{Token Removal} denotes strategies that resolve ambiguity by pruning candidate tokens, whereas \emph{Distribution Preserving} denotes methods that preserve the original generation distribution during disambiguation.

\textbf{Token-removal strategies cannot provide strict security.}
These methods consistently yield non-zero KL divergence, indicating unavoidable distortion of the original generation distribution. In particular, Basic \cite{nozaki2022addressing} causes a clear PPL shift, suggesting degraded text quality, while MWIS \cite{yan2023securedisambiguating} reduces this effect but still introduces non-negligible distributional distortion. Stepwise Verification \cite{yan2025addressingti,yan2025lowoverhead} further lowers KL divergence, but does not eliminate it; moreover, as the candidate-pool size increases, the number of candidate tokens that must be checked grows substantially, leading to higher time overhead.

\textbf{Distribution-preserving strategies still face efficiency trade-offs.}
Syncpool \cite{qi2025syncpool} maintains zero KL divergence through synchronized sampling over ambiguity pools, but does so at the cost of lower embedding capacity and entropy utilization. Verification-based \cite{yan2024nearimperceptible} performs well in security and capacity, but incurs substantially higher decoding time, which increases rapidly with the candidate-pool size.

\textbf{ReTokSync preserves both security and efficiency while ensuring reliable disambiguation.}
ReTokSync remains nearly identical to the baseline in both security and efficiency, while preventing local tokenization mismatches from escalating into complete extraction failure. Compared with the closest competitor Stepwise Verification, ReTokSync achieves strictly zero average and maximum KL divergence (vs.\ non-zero for Stepwise Verification in most settings), higher embedding capacity, and dramatically lower time overhead---for example, at $k{=}512$ on the English dataset, ReTokSync incurs only 1.74\% RTO compared with 525.29\% for Stepwise Verification. Compared with the distribution-preserving Syncpool, ReTokSync achieves substantially higher embedding capacity (e.g., 3.50 vs.\ 0.53 bits/token at $k{=}512$). As a result, ambiguity-induced errors are reduced to sparse and localized residual bit flips, yielding extraction accuracy above 99.7\%. Under the proposed two-channel communication mechanism, these residual errors can be fully corrected with a single additional sample, restoring system-level extraction accuracy to 100\%.

\subsection{Communication-Scenario Evaluation}

We further evaluate the proposed framework in a realistic communication scenario. Using the same backbone models and datasets as above, we employ ReTokSync as the primary communication mechanism and use Syncpool as an auxiliary reliable channel to transmit correction information. For each setting, we generate 500 steganographic samples of length 100 and group every 10 samples for joint correction. In the primary channel, the top-$k$ value is set to 32, 64, 128, 256, and 512, while the correction channel consistently uses Syncpool with top-$k=64$.

\begin{table}[t]
\caption{Grouped-correction performance under different top-$k$ settings.}
\label{tab:comm_perf_10}
\centering
\footnotesize
\setlength{\tabcolsep}{4pt}
\renewcommand{\arraystretch}{1.08}
\begin{tabular}{lcccc}
\toprule
\textbf{Model} & \textbf{Top-$k$} & \makecell[c]{\textbf{Avg. Error} \\ \textbf{(bit/group)}} & \makecell[c]{\textbf{Avg. Corr. Len.} \\ \textbf{(bit)}} & \makecell[c]{\textbf{Success} \\ \textbf{(\%)}} \\
\midrule
Llama-3.1-8B & 32  & 0.78 & 10.28 & 100 \\
Llama-3.1-8B & 64  & 1.40 & 14.82 & 100 \\
Llama-3.1-8B & 128 & 1.82 & 16.44 & 100 \\
Llama-3.1-8B & 256 & 2.06 & 17.98 & 100 \\
Llama-3.1-8B & 512 & 2.96 & 21.30 & 100 \\
\midrule
Qwen3-8B & 32  & 1.22 & 12.84 & 100 \\
Qwen3-8B & 64  & 1.52 & 14.22 & 100 \\
Qwen3-8B & 128 & 2.92 & 21.22 & 100 \\
Qwen3-8B & 256 & 5.22 & 30.16 & 100 \\
Qwen3-8B & 512 & 4.40 & 25.94 & 100 \\
\bottomrule
\end{tabular}
\end{table}

Table~\ref{tab:comm_perf_10} summarizes the group-level correction performance under different top-$k$ settings, including the average number of residual bit errors per group, the average correction-message length, and the final correction success rate. Table~\ref{tab:error-ratio-topk} complements this view by reporting the corresponding bit error ratio, token error ratio, and correction-to-embedding ratio.


Several consistent observations emerge. First, the correction success rate remains at \(100\%\) for all evaluated settings, confirming that the residual errors produced by the primary channel can be fully repaired under the proposed two-channel design. Second, the residual error remains small and the associated correction-message length is correspondingly limited; although the required correction length generally increases with the primary top-$k$, it stays well controlled across all settings. Third, Table~\ref{tab:error-ratio-topk} shows that these residual deviations remain highly sparse even in the more difficult configurations: the bit error ratio stays within 0.033\%--0.083\% for Llama-3.1-8B and 0.064\%--0.199\% for Qwen3-8B, while the correction-to-embedding ratio remains below 1\% in almost all settings. Additional experimental details are provided in the supplementary material.


These results confirm the complementarity of the two channels: ReTokSync keeps the residual deviations sufficiently sparse for the auxiliary Syncpool channel to repair them with a small amount of structured correction information. The primary channel preserves efficient and secure payload transmission, while the auxiliary channel converts sparse residual errors into fully recoverable system-level communication---a favorable operating point for continuous interaction scenarios.

\begin{table}[t]
\centering
\footnotesize
\caption{Residual error sparsity and correction overhead under different top-$k$ settings.}
\label{tab:error-ratio-topk}
\renewcommand{\arraystretch}{1.02}
\setlength{\tabcolsep}{3.0pt}
\begin{tabular}{lcccc}
\toprule
\textbf{Model} & \textbf{Top-$k$} & \makecell[c]{\textbf{Bit Error} \\ \textbf{Ratio}} & \makecell[c]{\textbf{Token Error} \\ \textbf{Ratio}} & \makecell[c]{\textbf{Correction-to-} \\ \textbf{Embedding Ratio}} \\
\midrule
\multirow{5}{*}{\makecell[c]{Llama-3.1-8B~\cite{grattafiori2024llama3herd}\\(English)}} 
& 32  & 0.033\% & 0.030\% & 0.435\% \\
& 64  & 0.052\% & 0.062\% & 0.547\% \\
& 128 & 0.061\% & 0.072\% & 0.551\% \\
& 256 & 0.063\% & 0.080\% & 0.549\% \\
& 512 & 0.083\% & 0.090\% & 0.596\% \\
\midrule
\multirow{5}{*}{\makecell[c]{Qwen3-8B~\cite{yang2025qwen3}\\(Chinese)}} 
& 32  & 0.064\% & 0.050\% & 0.671\% \\
& 64  & 0.070\% & 0.054\% & 0.658\% \\
& 128 & 0.120\% & 0.100\% & 0.875\% \\
& 256 & 0.199\% & 0.140\% & 1.148\% \\
& 512 & 0.155\% & 0.116\% & 0.912\% \\
\bottomrule
\end{tabular}
\end{table}

\subsection{Security Against Steganalysis}

For steganalysis evaluation, we consider three neural detectors with different architectural characteristics: FCN \cite{yang2019faststeganalysis}, BiLSTM-DENSE \cite{yang2020featurepyramid}, and RBILSTMC \cite{niu2019rbilstmc}. Each experiment is conducted on a balanced dataset consisting of 10,000 cover texts generated by Llama-3.1-8B and 10,000 stego texts. The dataset is split into training, validation, and test sets with a ratio of 60\%/20\%/20\%, and all detectors are trained for 20 epochs. Table~\ref{tab:steganalysis} reports the detection accuracies of the Baseline and ReTokSync under the three steganalysis models. Across all three detectors, the accuracies remain close to chance level (50\%), suggesting that ReTokSync preserves the security properties of the underlying steganographic scheme under the present evaluation protocol.

\begin{table}[t]
\centering
\small
\setlength{\tabcolsep}{6pt}
\caption{Steganalysis Accuracy.}
\label{tab:steganalysis}
\begin{tabular}{c c c}
\toprule
Scheme & Steganalysis Model Architecture & Detection Accuracy \\
\midrule
\multirow{3}{*}{Baseline}
& FCN~\cite{yang2019faststeganalysis} & 49.60\% \\
& BiLSTM-DENSE~\cite{yang2020featurepyramid} & 49.70\% \\
& RBILSTMC~\cite{niu2019rbilstmc} & 49.77\% \\
\midrule
\multirow{3}{*}{ReTokSync}
& FCN~\cite{yang2019faststeganalysis} & 49.47\% \\
& BiLSTM-DENSE~\cite{yang2020featurepyramid} & 50.40\% \\
& RBILSTMC~\cite{niu2019rbilstmc} & 49.42\% \\
\bottomrule
\end{tabular}
\end{table}

\section{Conclusion}


We studied tokenization ambiguity in generative linguistic steganography and proposed ReTokSync, a self-synchronizing framework that detects ambiguity online and performs corrective reset only when needed. By treating ambiguity as a locally triggered, inherently sparse phenomenon, ReTokSync prevents local tokenization mismatch from propagating into persistent desynchronization and global extraction failure while preserving zero KL divergence and near-baseline efficiency. Experiments on both English and Chinese settings confirm that ReTokSync remains closest to the baseline across all evaluated metrics, with residual bit error ratios below 0.2\%. Building on this property, we further developed a two-channel communication mechanism that fully repairs these residual errors, achieving 100\% end-to-end recovery. Future work will extend the framework to broader model families, alternative tokenization schemes, and stronger adversarial security evaluations.

\bibliographystyle{ACM-Reference-Format}
\bibliography{references}










\clearpage
\appendix

\section{Additional Details of ReTokSync}
\label{sec:appendix-retoksync}

\subsection{Skip-X Mechanism}
\label{sec:appendix-skipx}

When truncated sampling strategies such as top-$k$ and top-$p$ are used, tokenization ambiguity may cause the re-tokenized sequence to contain tokens that lie outside the support of the original truncated distribution \cite{holtzman2020curious,yan2025addressingti}. We refer
to such tokens as \emph{X-tokens}. More concretely, let
\[
\tilde{x}_{1:\tilde{t}} = \mathrm{Tok}(\mathrm{Detok}(x_{1:t}))
\]
denote the receiver-view token sequence defined in the main paper. For each
receiver-view step $u \in \{1,\dots,\tilde{t}\}$, let
\[
\mathcal{C}_u
\]
denote the truncated candidate set determined by the shared conditional
distribution
\[
p_\theta(\cdot \mid \tilde{x}_{1:u-1}).
\]
If $\tilde{x}_u \notin \mathcal{C}_u$, then $\tilde{x}_u$ is an X-token.

X-tokens are introduced by re-tokenization rather than by the sender's actual
sampling decision at the corresponding step. Therefore, directly applying the
standard extraction rule to such tokens may produce spurious message bits and
incorrect state transitions. To address this issue, we augment the decoding
procedure $\mathrm{Dec}$ with a simple \textbf{Skip-X} rule: whenever the
current token is an X-token, the decoder performs neither message extraction
nor state update, and directly proceeds to the next receiver-view step.

Formally, suppose ambiguity is detected at generation step $t$, and the sender
runs
\[
(\hat{m}_{1:\hat{j}}, \hat{s}) \leftarrow \mathrm{Dec}(\tilde{x}_{1:\tilde{t}}, s_0)
\]
on the current receiver-view sequence. Under Skip-X, if the token at receiver-view
step $u$ satisfies $\tilde{x}_u \notin \mathcal{C}_u$, then the decoder skips
this position and keeps the current extraction status unchanged. That is, no
new message bits are appended, and the internal state is preserved. Since the
receiver applies the same rule during extraction, both parties remain strictly
aligned after corrective reset.

This mechanism is particularly important when a locally ambiguous tokenization
introduces a non-original token. For example, two tokens \texttt{any} and
\texttt{thing} may be merged into a single token \texttt{anything} after
re-tokenization, while \texttt{anything} does not belong to the truncated
candidate set at the corresponding step. The small number of residual bit-flip
errors in the basic corrective-reset procedure originates precisely from such
non-original tokens: if they are decoded as ordinary in-support tokens, the
decoder may extract incorrect message bits. Skip-X explicitly identifies these
out-of-support tokens and bypasses them during extraction, thereby eliminating
this error source and restoring message extraction accuracy to 100\%.

\begin{algorithm}[t]
\caption{Decoding with Skip-X}
\label{alg:dec-skipx}
\begin{algorithmic}[1]
\REQUIRE Receiver-view sequence $\tilde{x}_{1:\tilde{t}}$, initial state $s_0$
\ENSURE Extracted message prefix $\hat{m}_{1:\hat{j}}$, final state $\hat{s}$
\STATE $\hat{m} \gets \emptyset$, $\hat{j} \gets 0$, $\hat{s} \gets s_0$
\FOR{$u = 1$ to $\tilde{t}$}
    \STATE Compute the truncated candidate set $C_u$ from $p_\theta(\cdot \mid \tilde{x}_{1:u-1})$
    \IF{$\tilde{x}_u \notin C_u$}
        \STATE \textbf{continue}
    \ENDIF
    \STATE $(\Delta m, \hat{s}) \gets \textsc{DecStep}(\tilde{x}_u, C_u, \hat{s})$
    \STATE $\hat{m} \gets \hat{m} \Vert \Delta m$
    \STATE $\hat{j} \gets \hat{j} + |\Delta m|$
\ENDFOR
\STATE \textbf{return} $(\hat{m}_{1:\hat{j}}, \hat{s})$
\end{algorithmic}
\end{algorithm}


\subsection{Incomplete Tokens and Buffering Mechanism}
\label{sec:appendix-incomplete-buffer}

Prior work has observed a class of ``not-yet-complete'' token phenomena from different perspectives . Land and Bartolo point out that some tokenizers produce \emph{partial UTF-8 sequences}, which contain only a fragment of the UTF-8 encoding of a character and therefore cannot be independently decoded into normal text \cite{land2024fishing}. They also discuss \emph{unreachable tokens}, i.e., tokens that cannot return to their original token IDs after a decode--re-encode process \cite{land2024fishing}. Building on this line of observation, Yan and Murawaki further show that such tokens often manifest themselves as \emph{temporary inconsistency} in generative steganography and watermarking \cite{yan2025addressingti}. Although the token-level inconsistency rate on 100-token texts is usually below 0.5\%, the temporary inconsistency rates on Swallow-7b and Qwen2.5-7b still reach 81.98\% and 87.93\%, respectively, and most of these inconsistencies naturally disappear within the next two generated tokens \cite{yan2025addressingti,fujii2024continual,qwenteam2025qwen25}. In other words, some tokens correspond only to locally unfinished encoding units at the time they are generated, and become normal text only after being combined with subsequent tokens.

Motivated by this observation, we refer to tokens that cannot be independently decoded into normal text and require subsequent tokens to form a complete textual unit as \emph{incomplete tokens}. In ReTokSync, since re-tokenization and ambiguity detection are performed after each generated token, directly re-tokenizing an unfinished combination of incomplete tokens may produce a mismatch between sequences that does not correspond to genuine tokenization ambiguity, but merely reflects a temporary inconsistency caused by incomplete generation.

Although Yan and Murawaki provide a systematic analysis of temporary inconsistency \cite{yan2025addressingti}, their steganographic Stepwise Verification does not introduce a dedicated treatment for such tokens. Their method is essentially a persistent preventive intervention: at each generation step, candidate tokens are verified one by one, and those judged inconsistent are removed in advance to avoid tokenization inconsistency at the source \cite{yan2025addressingti}. While this design guarantees consistency, it does not explicitly distinguish genuine tokenization ambiguity from temporary inconsistency caused by unfinished token fragments.

To handle incomplete tokens, we design a buffering mechanism within ReTokSync. Specifically, incomplete tokens are first accumulated in a buffer, and ambiguity detection is triggered only when the buffered token sequence becomes decodable as normal text. For incomplete-token combinations that have not yet been fully formed, the system temporarily skips ambiguity detection and related processing, thereby avoiding misclassifying the sequence offset caused by incomplete decoding as genuine ambiguity. Through this buffering mechanism, ambiguity detection avoids both false positives and false negatives, thereby preserving the precision of ReTokSync in handling tokenization ambiguity.

\begin{algorithm}[t]
\caption{Incomplete-Token Buffering in ReTokSync}
\label{alg:incomplete-buffering}
\begin{algorithmic}[1]
\REQUIRE Payload bitstring $m$, tokenizer $\textsc{Tok}/\textsc{Detok}$, stego encoder $\textsc{Enc}$, stego decoder $\textsc{Dec}$, predicates $\textsc{IsIncomplete}(\cdot)$ and $\textsc{Decodable}(\cdot)$, initial state $s_0$
\STATE $x \gets [\,]$, $\tilde{x} \gets [\,]$, $B \gets [\,]$, $j \gets 0$, $s \gets s_0$
\FOR{$t = 1, 2, \ldots$}
    \STATE $(x_t, j, s) \gets \textsc{Enc}(\tilde{x}, m, j, s)$
    \STATE $x \gets x \Vert [x_t]$
    \STATE $\tilde{x}_{\mathrm{pred}} \gets \tilde{x} \Vert [x_t]$
    \IF{$B \neq [\,]$ or $\textsc{IsIncomplete}(x_t)$}
        \STATE $B \gets B \Vert [x_t]$
        \IF{not $\textsc{Decodable}(B)$}
            \STATE $\tilde{x} \gets \tilde{x}_{\mathrm{pred}}$
            \STATE \textbf{continue}
        \ENDIF
        \STATE $B \gets [\,]$
    \ENDIF
    \STATE $\tilde{x}_{\mathrm{retok}} \gets \textsc{Tok}(\textsc{Detok}(x))$
    \IF{$\tilde{x}_{\mathrm{pred}} \neq \tilde{x}_{\mathrm{retok}}$}
        \STATE $(\hat{m}_{1:\hat{j}}, \hat{s}) \gets \textsc{Dec}(\tilde{x}_{\mathrm{retok}}, s_0)$
        \STATE $j \gets \hat{j}$, $s \gets \hat{s}$
    \ENDIF
    \STATE $\tilde{x} \gets \tilde{x}_{\mathrm{retok}}$
\ENDFOR
\end{algorithmic}
\end{algorithm}

\subsection{Handling Pre-segmentation-Induced Temporary Inconsistency}
\label{sec:appendix-preseg}

Before subword encoding, a tokenizer often performs a pre-segmentation
(or pre-tokenization) step that partitions the input according to
hand-crafted rules, such as whitespace, punctuation, or newline-related
markers \cite{schmidt2024tokenization}.
This behavior is also exhibited by the tokenizers used in model families
such as GPT2 \cite{radford2019languagemodels} and Llama3
\cite{grattafiori2024llama3herd}. Once such pre-segmentation boundaries
are fixed, subsequent merges are restricted to operate within, rather than
across, the resulting segments. Consequently, even if two adjacent tokens
are in principle mergeable, they will not be merged when they fall on
different sides of a segment boundary, for example, when one lies at the
end of the preceding segment and the other at the beginning of the next
segment. Therefore, whether two adjacent tokens are merged depends not only
on their local adjacency, but also on whether the relevant
pre-segmentation boundary has already been established.

This phenomenon is particularly important in ReTokSync. Since ReTokSync
performs re-tokenization and ambiguity detection after each generated token,
the input to the tokenizer is often only an unfinished text prefix. As a
result, the resulting pseudo sequence is not always the strict receiver-view
token sequence. If the current prefix happens to fall into a pre-segmentation
region whose boundary has not yet stabilized, the resulting sequence mismatch
does not correspond to genuine tokenization ambiguity, but merely reflects a
temporary inconsistency caused by an unfinished pre-segmentation structure. If
detection and correction are triggered at this stage, normal ambiguity-free
generation positions may be falsely treated as ambiguous, which contradicts
the design goal of ReTokSync to intervene only when genuine ambiguity occurs.

To address this issue, we design a processing scheme based on a special token
family. Concretely, through empirical analysis, we identify a set of special
tokens that act as pre-segmentation anchors in the tokenizers of GPT2 and the
Llama3 series, and organize them into a family $\mathcal{F}$. The treatment is
parallel to the buffering mechanism for incomplete tokens: when the newly
generated token belongs to $\mathcal{F}$, the system temporarily skips
re-tokenization, ambiguity detection, and corrective reset; instead, it
directly appends the token to the current pseudo sequence and continues
subsequent generation with this pseudo sequence as context. Only after the
first token outside $\mathcal{F}$ is generated does the system resume the
standard re-tokenization and disambiguation procedure.

The reason for this design is that future tokens are unknown at the current
step, and thus it is impossible to determine in advance whether the
pre-segmentation boundary has been fully established. This also explains why
existing methods based on preventive intervention cannot specially handle this
issue. Only after the first token outside $\mathcal{F}$ appears does the
relevant pre-segmentation structure become fixed, making it possible to decide
the actual segmentation pattern of the current prefix in the complete text and,
accordingly, whether genuine tokenization ambiguity has occurred. Experiments
show that this mechanism strictly avoids falsely intervening on normal
ambiguity-free positions, thereby preserving the precision of ReTokSync in
pre-segmentation scenarios.

\begin{algorithm}[t]
\caption{Deferring Detection under Pre-segmentation Anchors}
\label{alg:preseg-deferral}
\begin{algorithmic}[1]
\REQUIRE Family $F$, payload bitstring $m$, tokenizer $\textsc{Tok}/\textsc{Detok}$, stego encoder $\textsc{Enc}$, stego decoder $\textsc{Dec}$, initial state $s_0$
\STATE $x \gets [\,]$, $\tilde{x} \gets [\,]$, $j \gets 0$, $s \gets s_0$
\FOR{$t = 1, 2, \ldots$}
    \STATE $(x_t, j, s) \gets \textsc{Enc}(\tilde{x}, m, j, s)$
    \STATE $x \gets x \Vert [x_t]$
    \IF{$x_t \in F$}
        \STATE $\tilde{x} \gets \tilde{x} \Vert [x_t]$
        \STATE \textbf{continue}
    \ENDIF
    \STATE $\tilde{x}_{\mathrm{pred}} \gets \tilde{x} \Vert [x_t]$
    \STATE $\tilde{x}_{\mathrm{retok}} \gets \textsc{Tok}(\textsc{Detok}(x))$
    \IF{$\tilde{x}_{\mathrm{pred}} \neq \tilde{x}_{\mathrm{retok}}$}
        \STATE $(\hat{m}_{1:\hat{j}}, \hat{s}) \gets \textsc{Dec}(\tilde{x}_{\mathrm{retok}}, s_0)$
        \STATE $j \gets \hat{j}$, $s \gets \hat{s}$
    \ENDIF
    \STATE $\tilde{x} \gets \tilde{x}_{\mathrm{retok}}$
\ENDFOR
\end{algorithmic}
\end{algorithm}

\subsection{Correctness Analysis of Message Extraction in ReTokSync}

Although ReTokSync allows tokenization ambiguity to arise at local generation positions, this does not imply that the corresponding message extraction at such positions necessarily becomes entirely incorrect. In fact, the impact of local ambiguity on message extraction depends on the specific case and can be categorized as follows.

First, if the ambiguous token introduced at a given position does not appear in the corresponding candidate probability distribution, the \texttt{skipx} criterion is triggered, and the receiver performs no message extraction at that position. In this case, no additional erroneous bits are introduced.

Second, if the message extracted at that position happens to form a prefix of the message embedded before synchronization reset, then, although a local deviation appears to occur, it does not introduce erroneous bits and therefore does not cause substantive damage to the current message recovery result.

Finally, even when the extracted bits at that position are inconsistent with the original message, the resulting effect is typically limited to bit flips at a small number of positions in the original bit stream, rather than causing the entire extraction process to fail immediately. This further indicates that the primary harm of tokenization ambiguity does not lie in the appearance of one or a few incorrect bits at a local position. Rather, the fundamental problem is that an erroneous extraction process disrupts the subsequent message-state synchronization between the communicating parties, causing subsequent probability reconstruction and message extraction to deviate continuously from the true generation trajectory, and ultimately leading to complete communication failure.

\section{Compatibility Across Steganographic Algorithms}
\label{sec:appendix-compatibility}

\subsection{Generative Linguistic Steganography}
\label{sec:appendix-gls}

Generative linguistic steganography (GLS) has emerged as an important paradigm for covert communication. Unlike traditional text steganography, which typically relies on rule-based substitution or post hoc modification of an existing carrier, GLS embeds secret information directly into the generation process. By leveraging large language models to represent natural language distributions, GLS can produce stegotext that more closely resembles naturally generated text in terms of fluency, semantic coherence, and statistical properties.

Modern generative language models are typically built on autoregressive Transformer-based architectures \cite{vaswani2017attention,radford2019languagemodels}.Given a prefix context $x_{<t} = (x_1, x_2, \dots, x_{t-1})$, the model defines at step $t$ a conditional distribution over the vocabulary $\mathcal{V}$, denoted by $p_\theta(x_t \mid x_{<t})$. Text generation then proceeds by repeatedly selecting or sampling the next token from this distribution and appending it to the current prefix. In practice, inference may use deterministic decoding, such as greedy decoding, or stochastic decoding with truncation strategies such as top-$k$ or top-$p$ sampling. Because open-ended generation typically requires a balance between textual quality and diversity, practical decoding is often stochastic, which naturally provides an entropy source for information embedding.

GLS exploits this stepwise conditional distribution to map secret messages to token selections during generation, thereby enabling reversible embedding and extraction while preserving the underlying generation distribution as much as possible.

Representative GLS schemes include the following:

\begin{itemize}
    \item \textbf{Discop}: Ding et al.\ embed secret information by constructing multiple distribution copies and further improve embedding efficiency through a Huffman-tree-based design \cite{ding2023discop}.
    
    \item \textbf{Meteor}: Kaptchuk et al.\ propose a scheme based on interval reversibility, where token sampling proceeds in a manner analogous to arithmetic coding \cite{kaptchuk2021meteor}.
    
    \item \textbf{AC-based Steganography}: Ziegler et al.\ introduce arithmetic coding into neural text generation and determine each token progressively according to the model's conditional distribution \cite{ziegler2019neural}.
    
    \item \textbf{ADG}: Zhang et al.\ propose Adaptive Dynamic Grouping, which recursively partitions the candidate set according to token probabilities for message embedding \cite{zhang2021provably}.
    
    \item \textbf{Shimmer}: Bai et al.\ develop an entropy-collecting mechanism that captures and reuses residual entropy during embedding to improve payload capacity \cite{bai2025shimmer}.
    
    \item \textbf{Framework-based Scheme}: Liao et al.\ abstract provably secure generative steganography into three components: Probability Recombination, Bin Sampling, and Uniform Steganography \cite{liao2025framework}.
\end{itemize}

\subsection{Algorithm-Specific Adaptation of ReTokSync}
\label{sec:appendix-alg-adaptation}

ReTokSync is broadly compatible with existing generative steganographic schemes. Its deployment nevertheless requires modest algorithm-specific adaptations to accommodate the internal states maintained by different methods. These adaptations are imple\-men\-ta\-tion-level rather than conceptual: they do not change the core embedding logic, the underlying sampling mechanism, or the intended effectiveness of the original schemes. Instead, they expose and synchronize the states necessary for ReTokSync to support online ambiguity detection and corrective reset while preserving the original design of each method.

\textbf{Discop, Meteor, and ADG require only lightweight state synchronization.} In these schemes, pseudorandomness is used to support distribution-copy selection, XOR masking, or synchronized sampling, respectively \cite{ding2023discop,kaptchuk2021meteor,zhang2021provably}. Moreover, previously embedded bits do not recursively alter how subsequent bits are embedded beyond the maintained message position and pseudorandom state. Consequently, when ReTokSync performs corrective reset, it is sufficient to realign the sender's message pointer and pseudorandom state with the decoding result under the receiver-view tokenization. This lightweight reset is enough to restore synchronization without modifying the original embedding procedure, making these methods naturally compatible with ReTokSync.

\textbf{The Framework-based Scheme requires synchronization of multiple coordinated random states.} This scheme combines token binning with uniform steganography and therefore maintains more than one source of internal randomness \cite{liao2025framework}. To integrate it with ReTokSync, corrective reset must recover not only the message position, but also the pseudorandom states associated with bin selection and with the in-bin embedding process. In addition, a special case arises when ambiguity introduces a non-original token that still lies within the truncated candidate set but falls outside the currently selected bin. In such cases, the Skip-X rule must also be invoked so that extraction skips the incompatible position and rolls back the relevant state consistently with the receiver's view.

\textbf{AC-based steganography additionally requires restoration of the arithmetic-coding interval.} In arithmetic-coding-based schemes, the decoder state is determined not only by the current message position but also by the evolving coding interval \cite{ziegler2019neural}. Therefore, corrective reset under ReTokSync must synchronize both quantities to the receiver-view decoding result. Without interval restoration, subsequent decoding would proceed from an inconsistent internal state even if the message pointer itself were correctly aligned.

\textbf{Shimmer requires synchronization of additional iterative intermediate states.} Beyond the message position and pseudorandom state, Shimmer maintains internal variables associated with its iterative interval-update process \cite{bai2025shimmer}. To ensure that corrective reset returns exactly to the embedding trajectory implied by the receiver-view tokenization, ReTokSync must record and restore these intermediate states as well. Importantly, this adaptation does not modify Shimmer's protocol logic or its original security argument; it only extends the set of internal variables that must be synchronized after ambiguity is detected.

\section{Security Analysis}
\label{sec:security}

\paragraph{Threat model and security goal.}
We consider the standard black-box observation model for generative
linguistic steganography. The adversary observes only the visible text
transmitted over the public channel. The adversary does not observe the
sender's internal token sequence $x$, the re-tokenized receiver-view
sequence $\tilde{x}$, the ambiguity flag, the message pointer $j$, the
internal state $s$, or any other hidden computation performed by the
sender. In particular, the public channel reveals only the visible text
prefix
\begin{equation}
y_t = Detok(x_{1:t}).
\label{eq:visible-prefix}
\end{equation}

Our goal is to show that ReTokSync preserves the security of the
underlying distribution-preserving steganographic encoder. Specifically,
we prove that, under the above threat model, ReTokSync induces exactly
the same distribution on visible texts as the natural generation channel,
and is therefore strictly imperceptible on the visible-text channel.

\paragraph{Natural visible-text channel.}
Let $\mathcal{V}$ be the model vocabulary, and let
$p_\theta(\cdot \mid c)$ denote the model's next-token distribution under
token context $c$. For any visible text prefix $y$, define the one-step
visible update map
\begin{equation}
\Phi(y,v) := Detok(Tok(y)\,\|\, [v]),
\label{eq:phi}
\end{equation}
where $v \in \mathcal{V}$ and $\|$ denotes concatenation.

The natural generation channel on visible texts is therefore
\begin{equation}
Q_{\mathrm{nat}}(y' \mid y)
=
\sum_{v:\,\Phi(y,v)=y'} p_\theta(v \mid Tok(y)).
\label{eq:qnat}
\end{equation}

\paragraph{Assumption on the underlying steganographic encoder.}
Let $Enc$ denote the underlying steganographic encoder used inside
ReTokSync. We assume that it is distribution preserving in the usual
sense: for every token context $c$, payload $m$, message pointer $j$,
and internal state $s$, the next token output by $Enc$ satisfies
\begin{equation}
\Pr[Enc(c,m,j,s)=v] = p_\theta(v \mid c),
\qquad \forall v \in \mathcal{V}.
\label{eq:enc-dist-preserve}
\end{equation}
That is, the encoder may use the secret payload to choose among tokens,
but its marginal output distribution remains identical to the model
distribution.

\begin{lemma}[Receiver-view context invariant]
\label{lem:context-invariant}
At the beginning of every generation step $t$ in ReTokSync, the sender's
maintained context $\tilde{x}_{t-1}$ is exactly the receiver-view
tokenization of the current visible prefix:
\begin{equation}
\tilde{x}_{t-1} = Tok(y_{t-1}),
\qquad
y_{t-1}=Detok(x_{1:t-1}).
\label{eq:context-invariant}
\end{equation}
\end{lemma}

\begin{proof}
We prove by induction on $t$.

For $t=1$, the initialization gives $x=\emptyset$ and
$\tilde{x}=\emptyset$, hence
\begin{equation}
\tilde{x}_0 = Tok(\emptyset)=\emptyset.
\label{eq:base-case}
\end{equation}

Assume that Eq.~\eqref{eq:context-invariant} holds at the beginning of
step $t$. After the sender generates a token $x_t$, ReTokSync computes
\begin{equation}
\tilde{x}_{\mathrm{retok}}=Tok(Detok(x_{1:t})).
\label{eq:retok-seq}
\end{equation}
Regardless of whether ambiguity is detected, the algorithm finally sets
\begin{equation}
\tilde{x}_t \leftarrow \tilde{x}_{\mathrm{retok}}.
\label{eq:update-xtilde}
\end{equation}
Therefore,
\begin{equation}
\tilde{x}_t = Tok(Detok(x_{1:t})) = Tok(y_t).
\label{eq:induction-step}
\end{equation}
Hence Eq.~\eqref{eq:context-invariant} also holds at the beginning of
step $t+1$.
\end{proof}

\begin{lemma}[One-step visible distribution preservation]
\label{lem:one-step}
Conditioned on any current visible prefix $y$, the next visible text
produced by ReTokSync follows exactly the natural visible-text channel:
\begin{equation}
\Pr[Y_t = y' \mid Y_{t-1}=y]
=
Q_{\mathrm{nat}}(y' \mid y).
\label{eq:one-step}
\end{equation}
\end{lemma}

\begin{proof}
Fix any visible prefix $y$. By Lemma~\ref{lem:context-invariant}, at the
beginning of the current step the sender uses
\begin{equation}
\tilde{x}=Tok(y)
\label{eq:tok-y}
\end{equation}
as its token context. By the distribution-preserving property in
Eq.~\eqref{eq:enc-dist-preserve}, the generated token $X_t$ satisfies
\begin{equation}
\Pr[X_t=v \mid Y_{t-1}=y] = p_\theta(v \mid Tok(y)).
\label{eq:xt-dist}
\end{equation}

The visible text emitted after generating $X_t$ is
\begin{equation}
Y_t = \Phi(y,X_t)=Detok(Tok(y)\,\|\, [X_t]).
\label{eq:yt-def}
\end{equation}
Therefore, for any visible output $y'$,
\begin{equation}
\begin{aligned}
\Pr[Y_t=y' \mid Y_{t-1}=y]
&=
\sum_{v:\,\Phi(y,v)=y'}
\Pr[X_t=v \mid Y_{t-1}=y] \\
&=
\sum_{v:\,\Phi(y,v)=y'}
p_\theta(v \mid Tok(y)) \\
&=
Q_{\mathrm{nat}}(y' \mid y).
\end{aligned}
\label{eq:one-step-derivation}
\end{equation}
This proves the claim.
\end{proof}

\begin{theorem}[Strict security of ReTokSync on the main channel]
\label{thm:main-security}
Under the threat model above, if the underlying steganographic encoder
is distribution preserving, then for any generation horizon $T$,
ReTokSync induces exactly the same distribution on visible-text
transcripts as the natural generation channel. Namely, for every visible
transcript $y_{1:T}$,
\begin{equation}
\Pr[\mathsf{ReTokSync}=y_{1:T}]
=
\Pr[\mathsf{Natural}=y_{1:T}].
\label{eq:transcript-eq}
\end{equation}
Consequently, for any adversary $\mathcal{A}$ that observes only the
visible text transcript,
\begin{equation}
\mathrm{Adv}^{\mathrm{sec}}_{\mathsf{ReTokSync}}(\mathcal{A}) = 0.
\label{eq:adv-zero}
\end{equation}
\end{theorem}

\begin{proof}
By the chain rule of probability, for any visible transcript $y_{1:T}$,
\begin{equation}
\Pr[\mathsf{ReTokSync}=y_{1:T}]
=
\prod_{t=1}^{T}
\Pr[Y_t=y_t \mid Y_{1:t-1}=y_{1:t-1}].
\label{eq:chain-rule}
\end{equation}
Because the channel is autoregressive, conditioning on the whole visible
history is equivalent to conditioning on the current visible prefix
$y_{t-1}$. By Lemma~\ref{lem:one-step},
\begin{equation}
\Pr[Y_t=y_t \mid Y_{1:t-1}=y_{1:t-1}]
=
Q_{\mathrm{nat}}(y_t \mid y_{t-1}).
\label{eq:step-qnat}
\end{equation}
Substituting Eq.~\eqref{eq:step-qnat} into
Eq.~\eqref{eq:chain-rule} yields
\begin{equation}
\Pr[\mathsf{ReTokSync}=y_{1:T}]
=
\prod_{t=1}^{T} Q_{\mathrm{nat}}(y_t \mid y_{t-1})
=
\Pr[\mathsf{Natural}=y_{1:T}].
\label{eq:transcript-product}
\end{equation}
Hence the two channels are identically distributed on all observable
transcripts, which proves Eq.~\eqref{eq:adv-zero}.
\end{proof}

\paragraph{Why corrective reset does not break security.}
The key reason is that corrective reset changes only the sender's hidden
internal state, not the public transcript. More precisely:
\begin{enumerate}
    \item the visible token/text for the current step has already been
    sampled and emitted before reset is performed;
    \item reset only updates $(j,s)$ so that subsequent embedding
    remains aligned with the receiver-view tokenization;
    \item by Lemma~\ref{lem:context-invariant}, after reset the sender
    again conditions future generation on $Tok(y_t)$, which is exactly
    the context used by the natural channel after observing the visible
    prefix $y_t$.
\end{enumerate}
Therefore, corrective reset affects only hidden synchronization
variables, while the observable visible-text distribution remains
unchanged.

\begin{corollary}[Zero KL divergence on the visible-text channel]
\label{cor:kl-zero}
Under the assumptions of Theorem~\ref{thm:main-security}, the KL
divergence between the visible-text distribution induced by ReTokSync
and that induced by the natural channel is exactly zero:
\begin{equation}
D_{\mathrm{KL}}
\bigl(
\mathsf{ReTokSync}
\;\|\;
\mathsf{Natural}
\bigr)
=0.
\label{eq:kl-zero}
\end{equation}
\end{corollary}

\begin{proof}
The result follows immediately from
Eq.~\eqref{eq:transcript-eq}.
\end{proof}

\paragraph{Remark on side channels.}
The above theorem is proved under the text-only observation model. It
does not cover extra observables outside the visible transcript, such as
timing, computation latency, packetization artifacts, or a correction
schedule that depends on the secret payload or on observed errors.
These factors are outside the present security definition and must be
fixed or separately masked in deployment.

\begin{corollary}[Compositional security of the two-channel mechanism]
\label{cor:composition}
Consider the continuous-interaction mechanism in which ReTokSync is used
as the primary channel and Syncpool is used as the auxiliary correction
channel. If
\begin{enumerate}
    \item the primary ReTokSync channel satisfies
    Theorem~\ref{thm:main-security};
    \item the auxiliary channel is itself strictly secure; and
    \item the positions at which auxiliary correction samples are
    inserted are determined by a public, payload-independent rule, e.g.,
    one correction sample after every fixed-size group,
\end{enumerate}
then the composed two-channel communication mechanism is strictly secure
on the visible-text channel.
\end{corollary}

\begin{proof}
Under condition (3), the interleaving pattern of primary and auxiliary
samples is fixed independently of the secret message, so it reveals no
additional information. Under conditions (1) and (2), each individual
channel is indistinguishable from its corresponding natural channel.
Therefore, by standard composition, the joint transcript distribution is
also indistinguishable from the corresponding natural mixed transcript.
\end{proof}

\section{Details of the Covert Communication Mechanism}
\label{app:covert-communication-details}

\subsection{Protocol Scope}
\label{app:protocol-scope}

The main text presents the high-level motivation and workflow of the two-channel communication mechanism. Here we focus on the protocol details needed for implementation, namely how residual errors are aggregated on the sender side, how the corresponding correction message is constructed, and how the receiver parses and applies that message to recover the final payload.

We use ReTokSync as the primary channel and Syncpool as the auxiliary correction channel \cite{qi2025syncpool}. The rationale is that ReTokSync confines ambiguity-induced damage to sparse local residual errors, while Syncpool provides a reliable vehicle for transmitting the comparatively short correction message. As a result, the primary channel carries the main payload with high efficiency, and the auxiliary channel is used only to repair the small number of residual extraction errors that remain after primary-channel decoding.

\subsection{Sender-Side Construction}
\label{app:sender-side-construction}

We adopt a grouped correction strategy. After transmitting \(n\) steganographic samples through the primary channel, the sender generates one additional correction sample for that group through the auxiliary channel. This correction sample summarizes all residual errors detected within the group.

For a given group, the sender first converts each transmitted sample to its receiver-view token sequence and concatenates these sequences in transmission order, thereby forming a single group-level token sequence. During this process, the sender records, for each token position, the message fragment that would be extracted by the receiver under the standard decoding procedure. Because the sender also knows the intended payload, it can align the receiver-view extracted stream with the target stream at the token level and identify the token positions at which the extracted fragment is inconsistent with the intended one.

Each detected mismatch gives rise to one correction item. A correction item specifies (i) the position of the token whose extracted fragment is erroneous and (ii) the replacement fragment that should be used at that position. Collecting all such items over the group yields the correction message to be embedded in the auxiliary channel.

\subsection{Receiver-Side Recovery}
\label{app:receiver-side-recovery}

On the receiver side, all primary-channel samples in the group are first decoded in the standard manner. The receiver likewise concatenates the corresponding receiver-view token sequences in transmission order and records the extracted message fragment associated with each token position. This produces a group-level extracted stream together with its token-level segmentation.

The receiver then decodes the correction sample transmitted through the auxiliary channel and parses the resulting correction message sequentially. For each correction item, the receiver locates the designated token position in the group-level token sequence and replaces the originally extracted fragment at that position with the replacement fragment carried in the correction message. After all correction items have been applied, the corrected group-level message stream is partitioned back into the message streams corresponding to the individual samples in the group, yielding the final recovered payload for each sample.

\subsection{Encoding of the Correction Message}
\label{app:encoding-correction-message}

Each correction message consists of a count field followed by a sequence of correction items. The count field is encoded with fixed length and specifies the number of correction items contained in the message. If this field is all zeros, the message indicates that no correction is required for the current group, and the receiver terminates correction processing immediately after reading the count field.

Each correction item contains two components: a position field and a replacement field. The position field identifies the token position to be corrected in the group-level token sequence. To reduce overhead, positions are encoded with dynamic length rather than with a uniform fixed-width representation. Concretely, after each correction item has been parsed, both sender and receiver update the range of token positions that remain to be described; the bit width of the next position field is then determined from that remaining range. This design avoids unnecessary positional redundancy, especially when only a small number of tokens in a long group require correction.

The replacement field carries the corrected message fragment for the designated token position. No explicit length field is needed for this component. Because the receiver has already recorded the token-level extraction result for the group, once a target token position is identified, the receiver also knows the length of the originally extracted fragment at that position. The replacement field can therefore be parsed using that known length, after which the erroneous fragment is directly substituted with the corrected one.

Overall, the correction message is parsed from left to right: the receiver first reads the count field, then iteratively parses each correction item by decoding its position field, determining the required replacement length from the original token-level extraction record, and reading the corresponding replacement bits. This process continues until all correction items specified by the count field have been consumed.


\section{Additional Experimental Results}
\subsection{ReTokSync across Different Steganographic Algorithms}
\label{app:exp-alg-compatibility}

To evaluate the generality of ReTokSync, we further examine its performance across a diverse set of generative steganographic algorithms in both English and Chinese settings. Specifically, we use Llama-3.1-8B \cite{grattafiori2024llama3herd} and Qwen3-8B \cite{yang2025qwen3} as the underlying language models and conduct experiments on the English and Chinese versions of the IMDB dataset \cite{maas2011learning}, respectively. Under top-$k=128$, we generate 500 texts of length 100 and compare each original algorithm (\emph{Baseline}) with its ReTokSync-enhanced counterpart from three perspectives: effectiveness, security, and efficiency. The evaluated methods span six representative families of generative steganography, namely Discop \cite{ding2023discop}, Meteor \cite{kaptchuk2021meteor}, AC-based Steganography \cite{ziegler2019neural}, ADG \cite{zhang2021provably}, Shimmer \cite{bai2025shimmer}, and the Framework-based Scheme \cite{liao2025framework}, together with their corresponding variants.

\begin{table*}[t]
\centering
\small
\setlength{\tabcolsep}{4pt}
\caption{Experimental Results of ReTokSync on Different Steganographic Algorithms with Llama-3.1-8B in the English Setting.}
\label{tab:compatibility-english}

\resizebox{\textwidth}{!}{
\begin{tabular}{c|c|c S S S S S S c S}
\toprule
\multicolumn{1}{c|}{Algorithm} &
\multicolumn{1}{c|}{Variant} &
\multicolumn{1}{c}{Scheme} &
\multicolumn{1}{c}{Ave PPL} &
\multicolumn{1}{c}{Ave KLD} &
\multicolumn{1}{c}{Max KLD} &
\multicolumn{1}{c}{Capacity} &
\multicolumn{1}{c}{Utilization} &
\multicolumn{1}{c}{Total Time} &
\multicolumn{1}{c}{RTO} &
\multicolumn{1}{c}{Accuracy} \\
\multicolumn{1}{c|}{} &
\multicolumn{1}{c|}{} &
\multicolumn{1}{c}{} &
\multicolumn{1}{c}{} &
\multicolumn{1}{c}{$(\times 10^{-2})$} &
\multicolumn{1}{c}{$(\times 10^{-2})$} &
\multicolumn{1}{c}{(bits/token)} &
\multicolumn{1}{c}{} &
\multicolumn{1}{c}{(s)} &
\multicolumn{1}{c}{(\%)} &
\multicolumn{1}{c}{} \\
\midrule
\multirow{4}{*}{Discop~\cite{ding2023discop}}
& \multirow{2}{*}{Base}
& Baseline  & 8.40 & 0.00 & 0.00 & 1.36 & 0.45 & 5.10 & 0.00 & \multicolumn{1}{c}{--} \\
&
& ReTokSync & 8.39 & 0.00 & 0.00 & 1.36 & 0.45 & 5.15 & 1.16 & 1.000 \\
\cmidrule(lr){2-11}
& \multirow{2}{*}{Huffman}
& Baseline  & 9.67 & 0.00 & 0.00 & 2.88 & 0.91 & 5.13 & 0.00 & \multicolumn{1}{c}{--} \\
&
& ReTokSync & 9.61 & 0.00 & 0.00 & 2.86 & 0.91 & 6.49 & 26.62 & 0.998 \\
\midrule

\multirow{4}{*}{Meteor~\cite{kaptchuk2021meteor}}
& \multirow{2}{*}{Meteor}
& Baseline  & 10.26 & 0.00 & 0.01 & 2.12 & 0.65 & 5.41 & 0.00 & \multicolumn{1}{c}{--} \\
&
& ReTokSync & 10.26 & 0.00 & 0.01 & 2.12 & 0.65 & 6.00 & 10.89 & 0.999 \\
\cmidrule(lr){2-11}
& \multirow{2}{*}{Meteor-R}
& Baseline  & 10.27 & 0.00 & 0.01 & 2.20 & 0.67 & 31.92 & 0.00 & \multicolumn{1}{c}{--} \\
&
& ReTokSync & 10.27 & 0.00 & 0.01 & 2.19 & 0.67 & 34.60 & 8.39 & 0.998 \\
\midrule

\multirow{2}{*}{AC~\cite{ziegler2019neural}}
& \multirow{2}{*}{AC}
& Baseline  & 10.00 & 0.01 & 0.31 & 3.20 & 0.99 & 2.52 & 0.00 & \multicolumn{1}{c}{--} \\
&
& ReTokSync & 10.20 & 0.01 & 0.31 & 3.22 & 0.99 & 2.77 & 9.92 & 0.998 \\
\midrule

\multirow{2}{*}{ADG~\cite{zhang2021provably}}
& \multirow{2}{*}{ADG}
& Baseline  & 9.47 & 0.48 & 3.04 & 2.01 & 0.64 & 5.07 & 0.00 & \multicolumn{1}{c}{--} \\
&
& ReTokSync & 9.21 & 0.47 & 3.07 & 1.98 & 0.64 & 5.79 & 14.20 & 0.999 \\
\midrule

\multirow{4}{*}{Shimmer~\cite{bai2025shimmer}}
& \multirow{2}{*}{Shimmer}
& Baseline  & 10.28 & 0.00 & 0.00 & 2.71 & 0.81 & 4.99 & 0.00 & \multicolumn{1}{c}{--} \\
&
& ReTokSync & 10.30 & 0.00 & 0.00 & 2.73 & 0.82 & 5.28 & 5.81 & 0.999 \\
\cmidrule(lr){2-11}
& \multirow{2}{*}{Shimmer-R}
& Baseline  & 10.49 & 0.00 & 0.00 & 2.75 & 0.82 & 5.03 & 0.00 & \multicolumn{1}{c}{--} \\
&
& ReTokSync & 10.49 & 0.00 & 0.00 & 2.77 & 0.83 & 5.29 & 5.17 & 0.999 \\
\midrule

\multirow{6}{*}{Framework-based~\cite{liao2025framework}}
& \multirow{2}{*}{Differential}
& Baseline  & 11.66 & 0.00 & 0.00 & 2.36 & 0.69 & 5.10 & 0.00 & \multicolumn{1}{c}{--} \\
&
& ReTokSync & 11.66 & 0.00 & 0.00 & 2.36 & 0.69 & 5.11 & 0.20 & 0.999 \\
\cmidrule(lr){2-11}
& \multirow{2}{*}{Stability}
& Baseline  & 11.23 & 0.00 & 0.00 & 2.08 & 0.62 & 5.09 & 0.00 & \multicolumn{1}{c}{--} \\
&
& ReTokSync & 11.26 & 0.00 & 0.00 & 2.07 & 0.62 & 5.88 & 15.52 & 0.999 \\
\cmidrule(lr){2-11}
& \multirow{2}{*}{Binary}
& Baseline  & 11.23 & 0.00 & 0.00 & 1.88 & 0.55 & 5.81 & 0.00 & \multicolumn{1}{c}{--} \\
&
& ReTokSync & 11.05 & 0.00 & 0.00 & 1.87 & 0.55 & 6.52 & 12.22 & 0.999 \\
\bottomrule
\end{tabular}
}
\end{table*}

\begin{table*}[t]
\centering
\small
\setlength{\tabcolsep}{4pt}
\caption{Experimental Results of ReTokSync on Different Steganographic Algorithms with Qwen3-8B in the Chinese Setting.}
\label{tab:compatibility-chinese}

\resizebox{\textwidth}{!}{
\begin{tabular}{c|c|c S S S S S S c S}
\toprule
\multicolumn{1}{c|}{Algorithm} &
\multicolumn{1}{c|}{Variant} &
\multicolumn{1}{c}{Scheme} &
\multicolumn{1}{c}{Ave PPL} &
\multicolumn{1}{c}{Ave KLD} &
\multicolumn{1}{c}{Max KLD} &
\multicolumn{1}{c}{Capacity} &
\multicolumn{1}{c}{Utilization} &
\multicolumn{1}{c}{Total Time} &
\multicolumn{1}{c}{RTO} &
\multicolumn{1}{c}{Accuracy} \\
\multicolumn{1}{c|}{} &
\multicolumn{1}{c|}{} &
\multicolumn{1}{c}{} &
\multicolumn{1}{c}{} &
\multicolumn{1}{c}{$(\times 10^{-2})$} &
\multicolumn{1}{c}{$(\times 10^{-2})$} &
\multicolumn{1}{c}{(bits/token)} &
\multicolumn{1}{c}{} &
\multicolumn{1}{c}{(s)} &
\multicolumn{1}{c}{(\%)} &
\multicolumn{1}{c}{} \\
\midrule

\multirow{4}{*}{Discop~\cite{ding2023discop}}
& \multirow{2}{*}{Base}
& Baseline  & 6.89 & 0.00 & 0.00 & 1.16 & 0.44 & 7.07 & 0.00 & \multicolumn{1}{c}{--} \\
&
& ReTokSync & 6.89 & 0.00 & 0.00 & 1.17 & 0.44 & 7.56 & 6.93 & \multicolumn{1}{c}{--} \\
\cmidrule(lr){2-11}
& \multirow{2}{*}{Huffman}
& Baseline  & 7.70 & 0.00 & 0.00 & 2.45 & 0.90 & 7.07 & 0.00 & \multicolumn{1}{c}{--} \\
&
& ReTokSync & 7.75 & 0.00 & 0.00 & 2.47 & 0.90 & 7.54 & 6.61 & 0.998 \\
\midrule

\multirow{4}{*}{Meteor~\cite{kaptchuk2021meteor}}
& \multirow{2}{*}{Meteor}
& Baseline  & 7.26 & 0.00 & 0.01 & 1.69 & 0.63 & 7.82 & 0.00 & \multicolumn{1}{c}{--} \\
&
& ReTokSync & 7.02 & 0.00 & 0.01 & 1.66 & 0.63 & 8.47 & 8.22 & 0.998 \\
\cmidrule(lr){2-11}
& \multirow{2}{*}{Meteor-R}
& Baseline  & 7.53 & 0.00 & 0.02 & 1.81 & 0.66 & 53.35 & 0.00 & \multicolumn{1}{c}{--} \\
&
& ReTokSync & 7.52 & 0.00 & 0.02 & 1.82 & 0.66 & 53.60 & 0.47 & 0.998 \\
\midrule

\multirow{2}{*}{AC~\cite{ziegler2019neural}}
& \multirow{2}{*}{AC}
& Baseline  & 6.79 & 0.01 & 0.29 & 2.52 & 0.98 & 3.55 & 0.00 & \multicolumn{1}{c}{--} \\
&
& ReTokSync & 6.76 & 0.01 & 0.25 & 2.52 & 0.98 & 3.64 & 2.41 & 0.998 \\
\midrule

\multirow{2}{*}{ADG~\cite{zhang2021provably}}
& \multirow{2}{*}{ADG}
& Baseline  & 7.67 & 0.41 & 2.77 & 1.70 & 0.61 & 7.15 & 0.00 & \multicolumn{1}{c}{--} \\
&
& ReTokSync & 7.70 & 0.41 & 2.78 & 1.69 & 0.61 & 7.73 & 8.08 & 0.997 \\
\midrule

\multirow{4}{*}{Shimmer~\cite{bai2025shimmer}}
& \multirow{2}{*}{Shimmer}
& Baseline  & 6.72 & 0.00 & 0.00 & 2.04 & 0.76 & 7.09 & 0.00 & \multicolumn{1}{c}{--} \\
&
& ReTokSync & 6.63 & 0.00 & 0.00 & 2.01 & 0.75 & 7.28 & 2.68 & 0.995 \\
\cmidrule(lr){2-11}
& \multirow{2}{*}{Shimmer-R}
& Baseline  & 6.75 & 0.00 & 0.00 & 2.08 & 0.77 & 7.03 & 0.00 & \multicolumn{1}{c}{--} \\
&
& ReTokSync & 6.68 & 0.00 & 0.00 & 2.05 & 0.76 & 7.23 & 2.83 & 0.996 \\
\midrule

\multirow{6}{*}{Framework-based~\cite{liao2025framework}}
& \multirow{2}{*}{Differential}
& Baseline  & 7.33 & 0.00 & 0.00 & 1.79 & 0.69 & 7.25 & 0.00 & \multicolumn{1}{c}{--} \\
&
& ReTokSync & 7.24 & 0.00 & 0.00 & 1.78 & 0.69 & 10.51 & 45.05 & 0.999 \\
\cmidrule(lr){2-11}
& \multirow{2}{*}{Stability}
& Baseline  & 7.60 & 0.00 & 0.00 & 1.61 & 0.61 & 7.18 & 0.00 & \multicolumn{1}{c}{--} \\
&
& ReTokSync & 7.70 & 0.00 & 0.00 & 1.63 & 0.61 & 8.56 & 19.13 & 0.999 \\
\cmidrule(lr){2-11}
& \multirow{2}{*}{Binary}
& Baseline  & 7.58 & 0.00 & 0.00 & 1.43 & 0.53 & 7.84 & 0.00 & \multicolumn{1}{c}{--} \\
&
& ReTokSync & 7.60 & 0.00 & 0.00 & 1.43 & 0.53 & 10.98 & 40.18 & 1.000 \\
\bottomrule
\end{tabular}
}
\end{table*}

As shown in Tables~\ref{tab:compatibility-english} and~\ref{tab:compatibility-chinese}, ReTokSync remains broadly compatible with heterogeneous generative steganographic mechanisms. Although its deployment requires method-specific state synchronization, these adaptations do not alter the underlying embedding logic. Instead, they allow the original schemes to operate within the corrective-reset framework while largely preserving their original security and efficiency characteristics.

This trend can be observed consistently across different algorithm families. In the English setting, for example, the Discop-Base variant remains nearly unchanged after integration with ReTokSync: the average PPL changes only from 8.40 to 8.39, the embedding capacity remains 1.36 bits/token, and the total runtime increases only slightly from 5.10\,s to 5.15\,s, while extraction accuracy reaches 1.000. A similar pattern is observed for Meteor-R in the Chinese setting, where the total runtime changes only marginally from 53.35\,s to 53.60\,s and the extraction accuracy remains high at 0.998. Even in comparatively less favorable cases, the degradation remains limited in nature. For instance, for the Framework-based Differential variant with Qwen3-8B, ReTokSync preserves zero average and maximum KLD and achieves 0.999 extraction accuracy, although the runtime overhead rises to 45.05\%, indicating that the adaptation cost is method-dependent rather than stemming from any loss of distributional security.

Overall, the results suggest that ReTokSync transfers well across substantially different embedding mechanisms. Across the reported settings with available accuracy values, the ReTokSync-enhanced variants consistently achieve extraction accuracy of at least 99.5\%, while maintaining security and efficiency close to those of their corresponding baselines.

\subsection{Communication Performance under Different Group Sizes}
\label{sec:appendix-comm-performance}

To further evaluate the proposed two-channel communication mechanism, we conduct experiments under different group sizes in both English and Chinese settings. ReTokSync is used as the primary channel, while Syncpool serves as the auxiliary correction channel. We use Llama-3.1-8B and Qwen3-8B as the underlying language models and evaluate on the English and Chinese versions of the IMDB dataset, respectively. For each setting, we generate 100 steganographic samples of length 100. In the primary channel, the top-$k$ value is varied over $\{32, 64, 128, 256, 512\}$, while the auxiliary channel consistently uses Syncpool with top-$k=64$. We report results for group sizes of 5, 10, and 20. For the primary channel, we report the utilization and the average number of residual bit errors per group. For the auxiliary channel, we report the average correction-message length, the maximum correction-message length, the auxiliary-channel utilization, the ratio of correction-message length to the total embedded payload length, and the final correction success rate.

Several consistent observations emerge from Tables~\ref{tab:comm_perf_5},~\ref{tab:comm_perf_10}, and~\ref{tab:comm_perf_20}. First, the primary channel remains stable across different group sizes. For both Llama-3.1-8B and Qwen3-8B, the primary-channel utilization stays within a narrow range of 0.88--0.92 across all evaluated settings, indicating that grouped correction does not materially affect the payload efficiency of the main transmission. Second, the correction success rate remains at 100\% in all reported configurations, showing that the residual errors produced by the primary channel can be fully repaired under the proposed two-channel design.

\begin{table*}[t]
\caption{Communication Performance of the Two-Channel Mechanism under Different Top-$k$ Settings ($5$ Samples per Group)}
\label{tab:comm_perf_5}
\centering
\renewcommand{\arraystretch}{1.12}
\begin{tabular*}{\textwidth}{@{\extracolsep{\fill}}lcccccccc@{}}
\toprule
\multirow{2}{*}{\textbf{Models}}
& \multirow{2}{*}{\textbf{Top-$k$}}
& \multicolumn{2}{c}{\textbf{Primary Channel}}
& \multicolumn{5}{c}{\textbf{Auxiliary Channel}} \\
\cmidrule(lr){3-4} \cmidrule(lr){5-9}
&
& \makecell[t]{\textbf{Utilization}}
& \makecell[t]{\textbf{Ave Error} \\ \textbf{(bit)}}
& \makecell[t]{\textbf{Ave Corr. Len.} \\ \textbf{(bit)}}
& \makecell[t]{\textbf{Max. Corr. Len.} \\ \textbf{(bit)}}
& \makecell[t]{\textbf{Utilization}}
& \makecell[t]{\textbf{Corr./Emb. Ratio} \\ \textbf{(\%)}}
& \makecell[t]{\textbf{Success} \\ \textbf{(\%)}} \\
\midrule

\multirow{5}{*}{\makecell[c]{Llama-3.1-8B~\cite{grattafiori2024llama3herd}\\(English)}}
& 32  & 0.90 & 0.39 & 7.97  & 44 & 0.39 & 0.675 & 100 \\
& 64  & 0.90 & 0.70 & 10.17 & 57 & 0.39 & 0.751 & 100 \\
& 128 & 0.91 & 0.91 & 10.90 & 68 & 0.40 & 0.729 & 100 \\
& 256 & 0.91 & 1.03 & 11.57 & 45 & 0.43 & 0.705 & 100 \\
& 512 & 0.92 & 1.48 & 13.21 & 84 & 0.38 & 0.738 & 100 \\
\midrule

\multirow{5}{*}{\makecell[c]{Qwen3-8B~\cite{yang2025qwen3}\\(Chinese)}}
& 32  & 0.88 & 0.61 & 9.19  & 41 & 0.50 & 0.962 & 100 \\
& 64  & 0.89 & 0.76 & 9.80  & 41 & 0.59 & 0.907 & 100 \\
& 128 & 0.91 & 1.46 & 13.18 & 58 & 0.57 & 1.089 & 100 \\
& 256 & 0.91 & 2.61 & 17.53 & 72 & 0.61 & 1.333 & 100 \\
& 512 & 0.92 & 2.20 & 15.39 & 65 & 0.60 & 1.084 & 100 \\
\bottomrule
\end{tabular*}
\end{table*}

\begin{table*}[t]
\caption{Communication Performance of the Two-Channel Mechanism under Different Top-$k$ Settings ($10$ Samples per Group)}
\label{tab:comm_perf_10}
\centering
\renewcommand{\arraystretch}{1.12}
\begin{tabular*}{\textwidth}{@{\extracolsep{\fill}}lcccccccc@{}}
\toprule
\multirow{2}{*}{\textbf{Models}}
& \multirow{2}{*}{\textbf{Top-$k$}}
& \multicolumn{2}{c}{\textbf{Primary Channel}}
& \multicolumn{5}{c}{\textbf{Auxiliary Channel}} \\
\cmidrule(lr){3-4} \cmidrule(lr){5-9}
&
& \makecell[t]{\textbf{Utilization}}
& \makecell[t]{\textbf{Ave Error} \\ \textbf{(bit)}}
& \makecell[t]{\textbf{Ave Corr. Len.} \\ \textbf{(bit)}}
& \makecell[t]{\textbf{Max. Corr. Len.} \\ \textbf{(bit)}}
& \makecell[t]{\textbf{Utilization}}
& \makecell[t]{\textbf{Corr./Emb. Ratio} \\ \textbf{(\%)}}
& \makecell[t]{\textbf{Success} \\ \textbf{(\%)}} \\
\midrule

\multirow{5}{*}{\makecell[c]{Llama-3.1-8B~\cite{grattafiori2024llama3herd}\\(English)}}
& 32  & 0.90 & 0.78 & 10.28 & 47  & 0.39 & 0.436 & 100 \\
& 64  & 0.90 & 1.40 & 14.82 & 66  & 0.39 & 0.547 & 100 \\
& 128 & 0.91 & 1.82 & 16.44 & 76  & 0.36 & 0.550 & 100 \\
& 256 & 0.91 & 2.06 & 17.98 & 58  & 0.39 & 0.548 & 100 \\
& 512 & 0.92 & 2.96 & 21.30 & 109 & 0.38 & 0.595 & 100 \\
\midrule

\multirow{5}{*}{\makecell[c]{Qwen3-8B~\cite{yang2025qwen3}\\(Chinese)}}
& 32  & 0.88 & 1.22 & 12.84 & 60 & 0.57 & 0.672 & 100 \\
& 64  & 0.89 & 1.52 & 14.22 & 54 & 0.53 & 0.658 & 100 \\
& 128 & 0.91 & 2.92 & 21.22 & 74 & 0.61 & 0.877 & 100 \\
& 256 & 0.91 & 5.22 & 30.16 & 98 & 0.66 & 1.147 & 100 \\
& 512 & 0.92 & 4.40 & 25.94 & 71 & 0.60 & 0.913 & 100 \\
\bottomrule
\end{tabular*}
\end{table*}

\begin{table*}[t]
\caption{Communication Performance of the Two-Channel Mechanism under Different Top-$k$ Settings ($20$ Samples per Group)}
\label{tab:comm_perf_20}
\centering
\renewcommand{\arraystretch}{1.12}
\begin{tabular*}{\textwidth}{@{\extracolsep{\fill}}lcccccccc@{}}
\toprule
\multirow{2}{*}{\textbf{Models}}
& \multirow{2}{*}{\textbf{Top-$k$}}
& \multicolumn{2}{c}{\textbf{Primary Channel}}
& \multicolumn{5}{c}{\textbf{Auxiliary Channel}} \\
\cmidrule(lr){3-4} \cmidrule(lr){5-9}
&
& \makecell[t]{\textbf{Utilization}}
& \makecell[t]{\textbf{Ave Error} \\ \textbf{(bit)}}
& \makecell[t]{\textbf{Ave Corr. Len.} \\ \textbf{(bit)}}
& \makecell[t]{\textbf{Max. Corr. Len.} \\ \textbf{(bit)}}
& \makecell[t]{\textbf{Utilization}}
& \makecell[t]{\textbf{Corr./Emb. Ratio} \\ \textbf{(\%)}}
& \makecell[t]{\textbf{Success} \\ \textbf{(\%)}} \\
\midrule

\multirow{5}{*}{\makecell[c]{Llama-3.1-8B~\cite{grattafiori2024llama3herd}\\(English)}}
& 32  & 0.90 & 1.56 & 15.04 & 48  & 0.41 & 0.319 & 100 \\
& 64  & 0.90 & 2.80 & 24.88 & 67  & 0.39 & 0.459 & 100 \\
& 128 & 0.91 & 3.64 & 27.96 & 87  & 0.44 & 0.468 & 100 \\
& 256 & 0.91 & 4.12 & 31.68 & 83  & 0.43 & 0.483 & 100 \\
& 512 & 0.92 & 5.92 & 38.04 & 115 & 0.38 & 0.531 & 100 \\
\midrule

\multirow{5}{*}{\makecell[c]{Qwen3-8B~\cite{yang2025qwen3}\\(Chinese)}}
& 32  & 0.88 & 2.44 & 20.96 & 63  & 0.71 & 0.549 & 100 \\
& 64  & 0.89 & 3.04 & 23.36 & 59  & 0.66 & 0.541 & 100 \\
& 128 & 0.91 & 5.84 & 38.32 & 103 & 0.60 & 0.792 & 100 \\
& 256 & 0.91 & 10.44 & 56.56 & 112 & 0.70 & 1.075 & 100 \\
& 512 & 0.92 & 8.80 & 48.04 & 133 & 0.65 & 0.846 & 100 \\
\bottomrule
\end{tabular*}
\end{table*}

Third, although the absolute correction length generally increases with group size, the relative correction overhead remains small and often becomes more favorable when more samples are corrected jointly. For example, in the Llama-3.1-8B setting with top-$k=32$, the average correction-message length increases from 7.97 bits for groups of 5 to 10.28 bits for groups of 10 and 15.04 bits for groups of 20, while the correction-to-embedding ratio decreases from 0.675\% to 0.436\% and then to 0.319\%. A similar trend can be observed in the Chinese setting. For Qwen3-8B with top-$k=128$, the average correction-message length rises from 13.18 to 21.22 and 38.32 bits as the group size increases from 5 to 10 and 20, whereas the correction-to-embedding ratio decreases from 1.089\% to 0.877\% and 0.792\%, respectively. These results indicate that larger groups require longer correction messages in absolute terms, but allow the correction cost to be amortized over a larger primary payload.

At the same time, the residual errors in the primary channel remain limited in scale. Even in relatively challenging settings, the average number of residual bit errors per group remains modest compared with the total payload carried by the group. For instance, with Llama-3.1-8B and top-$k=512$, the average number of residual bit errors increases from 1.48 for groups of 5 to 2.96 for groups of 10 and 5.92 for groups of 20, while the corresponding correction-to-embedding ratio remains below 0.8\% in all three cases. For Qwen3-8B, the largest reported correction-to-embedding ratio is 1.333\% at top-$k=256$ with groups of 5, and it remains at or below 1.147\% for groups of 10 and 1.075\% for groups of 20. Taken together, these results show that the residual deviations produced by ReTokSync are highly sparse, so that they can be repaired using only a short auxiliary correction message.

Overall, the results confirm the complementarity of the two channels. ReTokSync provides an efficient primary channel whose residual errors remain sparse and well structured, while Syncpool provides a reliable auxiliary channel that repairs these errors with 100\% success and modest overhead. This makes the proposed mechanism a practical design for continuous covert communication, as it preserves high primary-channel efficiency while enabling exact end-to-end recovery through lightweight structured correction.

\subsection{Qualitative Stegotext Examples}
\label{app:qualitative-stegotext-examples}

To provide a qualitative view of the generated stegotext, we present representative examples in both English and Chinese. Specifically, the English example is generated with Llama-3.1-8B, and the Chinese example is generated with Qwen3-8B. For each case, we compare the stegotext produced by the baseline scheme with that generated after applying ReTokSync under the same input context. The examples are shown in Tables~\ref{tab:stegotext-example-english} and~\ref{tab:stegotext-example-chinese}.

These examples illustrate that ReTokSync does not introduce obvious qualitative degradation in the surface text. In both languages, the generated outputs remain fluent and semantically coherent, while still reflecting the variability inherent in stochastic steganographic generation. Since the English and Chinese examples are generated from different contexts, we present them in separate tables for clarity.

\begin{table*}[t]
\centering
\caption{Qualitative English stegotext example generated by Llama-3.1-8B.}
\label{tab:stegotext-example-english}
\small
\renewcommand{\arraystretch}{1.08}
\begin{tabular}{@{}p{\textwidth}@{}}
\toprule

\textbf{Context:} \\
\emph{I was excited to watch a Catalu\~na film in Berlin's competition. But after the screening, I was totally disappointed and furious: too much blood, too much time, too many themes for nothing.} \\

\midrule

\textbf{Stegotext (Discop\cite{ding2023discop}):} \\
\emph{Something like ``B\'endezinhao'' isn't it! That film has blood inside and is simply ``an animal''. Very good acting, but only four people are in focus; the others on screen stay in the background, like shadows. This afternoon at Cine France, the second film was the Mexican film ``El Hombre sin Sombra''. Incredible film---simply incredible.} \\

\midrule

\textbf{Stegotext (ReTokSync):} \\
\emph{Something like ``B\'endezinhao'' isn't it! That film has blood inside and is simply ``an animal''. Very good acting, but only four people are in focus; the others on screen stay in the background, like shadows. This afternoon at Cine France, the second film, the Mexican film ``El Hombre sin Sombra'', wasn't well executed: some very bad films and photo projections. I don't like it.} \\

\bottomrule
\end{tabular}
\end{table*}

\begin{table*}[t]
\centering
\caption{Qualitative Chinese stegotext example generated by Qwen3-8B.}
\label{tab:stegotext-example-chinese}
\small
\renewcommand{\arraystretch}{1.08}
\begin{CJK*}{UTF8}{gbsn}
\begin{tabular}{@{}p{\textwidth}@{}}
\toprule

\textbf{Context:} \\
这部电影从来没有在我们地区的影院上映，所以当它在DVD上出现时，我是第一个租借它的人之一。 \\

\midrule

\textbf{Stegotext (Discop\cite{ding2023discop}):} \\
但是，不是很令人满意。我觉得：主角扮演得不够好，但是最差的是，拍摄手法很抽劣。其次，电影的节奏太慢，没有必要地拖长时间，这导致了漫长的沉默，使人们感到很无聊。最后，它的结局很奇怪。我感觉这个电影没有达到预期，没有让我满意。 \\

\midrule

\textbf{Stegotext (ReTokSync):} \\
但是，不是很令人满意。我觉得：主角扮演得不够好，但是最差的是，拍摄手法很抽劣。其次，电影的节奏太慢，长途旅行的场景过于拖沓。唯一的亮点是男主角的表演，但我觉得他的表演过于做作，而且在某些场景里简直崩溃了。 \\

\bottomrule
\end{tabular}
\end{CJK*}
\end{table*}

\end{document}